\documentclass[lettersize,journal]{IEEEtran}
\usepackage{amsmath,amsfonts, amssymb}
\usepackage{algorithmic}
\usepackage{algorithm}
\usepackage{array}
\usepackage{bm}
\usepackage{textcomp}
\usepackage{stfloats}
\usepackage{url}
\usepackage{verbatim}
\usepackage{graphicx}
\usepackage{cite}
\usepackage{xcolor}
\usepackage[font=small,labelfont=bf]{caption}
\usepackage[font=small,labelfont=bf]{subcaption}
\usepackage{makecell}
\usepackage{multicol}
\usepackage{multirow}
\usepackage{booktabs}
\newtheorem{lemma}{Lemma}

\newenvironment{proof}{ \textbf{Proof:} }{ \hfill $\Box$}

\def\bb0{{\mathbb{0}}}

\def\ba{{\mathbf{a}}}
\def\bb{{\mathbf{b}}}

\def\bff{{\mathbf{f}}}
\def\bg{{\mathbf{g}}}
\def\bh{{\mathbf{h}}}

\def\b0{{\mathbf{0}}}
\def\bbC{{\mathbb{C}}}

\def\bbE{{\mathbb{E}}}

\def\bbN{{\mathbb{N}}}

\def\cG{\mathcal{G}}

\def\cO{\mathcal{O}}

\def\cS{\mathcal{S}}

\def\cU{\mathcal{U}}

\def\sfA{\mathsf{A}}
\def\sfB{\mathsf{B}}

\def\sfE{\mathsf{E}}
\def\sfF{\mathsf{F}}

\def\sfH{\mathsf{H}}

\def\sfM{\mathsf{M}}

\def\sfP{\mathsf{P}}
\def\sfQ{\mathsf{Q}}

\def\sfS{\mathsf{S}}
\def\sfT{\mathsf{T}}

\def\sfc{{\mathsf{c}}}
\def\sfd{{\mathsf{d}}}
\def\sfee{{\mathsf{e}}}

\def\sfg{{\mathsf{g}}}

\def\sfi{{\mathsf{i}}}
\def\sfj{{\mathsf{j}}}

\def\sfl{{\mathsf{l}}}
\def\sfm{{\mathsf{m}}}

\def\sfp{{\mathsf{p}}}

\def\sfr{{\mathsf{r}}}

\def\sft{{\mathsf{t}}}
\def\sfu{{\mathsf{u}}}

\def\sfw{{\mathsf{w}}}
\def\sfx{{\mathsf{x}}}
\def\sfy{{\mathsf{y}}}
\def\sfz{{\mathsf{z}}}
\def\sf0{{\mathsf{0}}}
\def\rm0{{\mathrm{0}}}
\def\b0{{\pmb{0}}} 
\begin{document}
	
	\title{	Frequency-selective beamforming and single-shot beam training  with dynamic metasurface antennas
	}

	\author{\IEEEauthorblockN{Nitish  Vikas Deshpande, \textit{Student Member, IEEE,} Joseph Carlson,  Miguel R. Castellanos, \textit{Member, IEEE,}  and Robert W. Heath Jr., \textit{Fellow, IEEE}}\\ 
		%	\IEEEauthorblockA{{Department of Electrical and Computer Engineering, North Carolina State University, Raleigh, NC} \\
			%		Email: \{nvdeshpa, mrcastel, rwheathjr\}@ncsu.edu}
		\thanks{Nitish Vikas Deshpande, Joseph Carlson,  and Robert W.  Heath Jr. are with the 
			Department of Electrical and Computer Engineering, University of California San Diego, La Jolla, CA 92093 USA (e-mail: nideshpande@ucsd.edu; j4carlson@ucsd.edu;  rwheathjr@ucsd.edu).
			Miguel R. Castellanos is with the Department of Electrical  Engineering and Computer Science at the University of Tennessee, Knoxville, TN 37996 USA (e-mail: mrcastellanos@utk.edu). This material is based upon work supported by the National Science Foundation under grant nos. NSF-ECCS-2435261, NSF ECCS-2414678, the Army Research Office under Grant W911NF2410107, and the Qualcomm Innovation Fellowship.}
	}

	\maketitle
	
	\begin{abstract}
		Dynamic metasurface antennas (DMAs) beamform through low-powered  components that enable  reconfiguration of each radiating element.
	Previous research on a single-user multiple-input-single-output (MISO) system with a dynamic metasurface antenna  at the transmitter has focused on maximizing the beamforming gain at a fixed operating frequency. 
		The DMA, however, has a frequency-selective response that leads to magnitude degradation for frequencies away from the resonant frequency of each element. 
	This causes reduction in beamforming gain if the DMA only operates at a fixed frequency.
	We exploit the frequency reconfigurability of the DMA to dynamically optimize both the operating frequency and the element configuration, maximizing the beamforming gain.
		    We leverage this approach to develop a single-shot beam training procedure using a DMA sub-array architecture that estimates the receiver's angular direction with a single OFDM pilot signal.
		    We evaluate the beamforming gain performance of the DMA array using the receiver's angular direction estimate obtained from  beam training. Our results show that it is sufficient to use a limited number of resonant frequency states to do both beam training and beamforming instead of using an infinite resolution DMA beamformer.    
	\end{abstract}
	
	
	\section{Introduction}\label{sec: Introduction }

	While most prior work on DMAs  assumes frequency-flat operation, their inherent ability to reconfigure beamforming properties over a wide bandwidth makes them promising candidates for frequency-selective wireless architectures~\cite{9324910,carlson_qif1}.
	A DMA consists of a waveguide with  multiple radiating slots that 
 can be independently  tuned to  resonate at different frequencies~\cite{boyarsky2021electronically,smith_analysis_2017}.
 Unlike conventional antenna arrays that require external components for beamforming, DMAs have the innate capability to beamform through  low-powered varactor or PIN diodes embedded into the antennas\cite{smith_analysis_2017}.
 Conventional analog components for frequency-selective beamforming such as  true time delay (TTD) elements~\cite{10051948}  and phase-shifters\cite{9735144} consume more power than varactor or PIN diodes\cite{rotman_true_2016}.
The low cost and low power consumption of DMAs make them highly versatile, with potential applications such as large array implementations in gNBs or use in battery-limited UE handsets.

DMAs present a unique set of challenges for beamforming compared to a traditional phased-array architecture because of the differences in their underlying operation mechanism.
The magnitude and phase response of each radiating slot in the DMA is coupled. The DMA architecture also has an intrinsic phase-advance in the waveguide.  
Furthermore, the underlying mechanism for beamforming in fact makes the beamforming frequency-selective. The operating bandwidth and extent of frequency selectivity are tightly coupled together in the design~\cite{carlson_qif1}.
As a result, even a task like configuring a DMA to produce a directive beam for a line-of-sight channel is deceptively complicated.

Prior work  on DMA signal processing is limited to  optimizing narrowband metrics.
The design of beamforming algorithms with DMAs is non-trivial because of the coupled magnitude and phase constraint, also referred to as the Lorentzian constraint in \cite{smith_analysis_2017}. One approach to configure DMA beamforming weights is to relax the Lorentzian constraint, solve for the desired phase-shift, and project the obtained weight on the set of feasible weights on the Lorentzian circle~\cite{9413746}. This relaxation and projection approach was shown to be sub-optimal in \cite{carlson2023hierarchical} which proposed another solution based on an optimal common phase rotation to each slot on the DMA.
  The single-user solution was generalized to a multi-user downlink system in \cite{9991243}, uplink system in~\cite{ShlezingerEtAlDynamicMetasurfaceAntennasUplink2019} and~\cite{9847609}, and
  hybrid beamforming with  the DMA in \cite{10096260} and \cite{10476911}.
The DMA beamforming approaches  in \cite{9413746, carlson2023hierarchical, 9991243, ShlezingerEtAlDynamicMetasurfaceAntennasUplink2019,9847609, 10096260, 10476911} treat the DMA beamforming weights as frequency-flat, i.e., the phase and magnitude response as a  function of frequency is not explicitly modeled.
Hence, this prior work on DMA is not directly applicable for wideband scenarios
where DMA frequency-selectivity and beam squint is significant. 
 Although frequency-reconfigurable antennas have been proposed for wireless systems with multi-band capabilities~\cite{4754001,sakkas2023frequency}, 
there is limited prior work on  analyzing frequency-selective  beamforming with DMAs. In \cite{WangEtAlDynamicMetasurfaceAntennasMIMOOFDM2021}, a DMA-based MIMO-OFDM system was analyzed by modeling the frequency-selectivity profile of the DMAs.  Their approach, however, does not compare frequency-selective DMA beamforming with traditional TTD or phased-array architectures.
In this paper, we bridge this gap by proposing a method to configure DMAs to achieve maximum beamforming gain comparable to TTD or phased-array systems.

In our paper,  we leverage the  ability of the DMA to be dynamically reconfigured over a wide frequency range to make the DMA performance comparable to the conventional architectures. We
envision a wireless system with DMA that has the ability to be reconfigured to different subcarrier frequencies.
	In current wireless standards, there exist methods that can dynamically adapt the operating frequency for several applications. For example, in 5G NR, techniques like dynamic spectrum sharing  allows the base station to dynamically switch between different  frequency bands depending on the network load and user demand~\cite{3gpp2018nr}.  Adaptive frequency hopping is used in Bluetooth to avoid interference, with the system able to operate at different frequency bands~\cite{5282356}.
Taking motivation from the current standards, 
our proposed method for DMA beamforming  allows the system operator to dynamically choose the DMA operating frequency within a feasible band. This enables configuring the DMA resonant frequencies based on the selected operating frequency.
We propose a two-stage optimization framework where the operating frequency is also an optimization variable. The first stage of the algorithm configures the DMA resonant frequencies to maximize the beamforming gain at a specific operating frequency, expressing the optimal resonant frequencies as a closed-form parametric function of the operating frequency. The second stage of the algorithm optimizes over all operating frequencies in the feasible operating band.
This approach is suitable for systems with spectrum availability similar to the DMA operating range, allowing for the selection of operating frequencies for communication. Our novel contribution is the second stage of optimizing the operating frequency, which mitigates the impact of magnitude degradation on the DMA beamforming weight. 
With this frequency-selective strategy, the beamforming gain and hence achievable rate  improves compared to prior work \cite{smith_analysis_2017,carlson2023hierarchical,9991243,9847609}. We further show that by optimal design choice for the DMA waveguide refractive index and inter-element spacing, the maximum beamforming gain for the DMA architecture equals that of a TTD array within the desired angular range.

To achieve the desired beamforming gain at the receiver, the transmit DMA array requires a good estimate of the receiver's angular direction.
 The second contribution of our paper is a beam training procedure for DMAs to estimate the receiver's angular direction.
 We demonstrate that with a DMA array, single-shot beam training is possible by leveraging the frequency reconfiguration property. Although single-shot beam training has been studied for TTD arrays \cite{9048885, 9349090} and a leaky waveguide antenna \cite{ghasempour2020single}, it has not been analyzed for DMAs in prior work. We propose a beam training procedure that probes different directions simultaneously using different subcarrier frequencies in a sub-array approach for the DMA array.

We summarize the key novelties of our paper with respect to prior work.				
We formulate the  frequency-selective DMA beamforming gain optimization for a line-of-sight channel. 		To make the maximum achievable DMA gain the same as that of TTD, we propose selecting the optimal operating frequency based on the receiver's angular direction unlike prior work which is limited to a single frequency optimization.		
			 We propose a single-shot beam training procedure for a  DMA planar array using a sub-array approach that enables the system to probe multiple angles simultaneously with a single OFDM symbol.
			We	compare the beamforming gain for two cases: (i) perfectly known receiver's angular direction, and (ii) angular direction estimated via the proposed training. We demonstrate that the second case, using a finite set of resonant frequencies, performs comparably to the first, which requires infinite frequency resolution.
				We show that the optimal operating frequency approach outperforms the benchmarks based on the fixed center frequency approach for all DMA operating frequency ranges. We also   obtain an achievable rate of  DMA that is comparable to a TTD array up to a certain bandwidth.

			Our paper is organized  as follows.
				We study a communication system where the transmitter is equipped with a DMA and the receiver has a single isotropic antenna. We discuss the single DMA transmitter model in Section~\ref{subsec: DMA as an array of reconfigurable polarizable dipoles}, formulate the propagation channel model and the received signal model in Section~\ref{subsec: Propagation channel model}, 
							and define the signal-to-noise ratio (SNR) and achievable rate in Section~\ref{sec: snr and bf gain}.
				In Section~\ref{sec: Frequency-selective beamforming gain optimization and analysis}, we analyze the beamforming gain for the single DMA transmitter.  In Section~\ref{sec: Single-shot beam training with DMAs}, we generalize the results of Section~\ref{sec: Frequency-selective beamforming gain optimization and analysis} to multiple DMAs  arranged as a planar array, propose the single-shot beam training procedure, and a beam training codebook design. In Section~\ref{subsec: Achievable rate results}, we evaluate the achievable rate DMA performance and compare with existing approaches. The MATLAB simulation code to generate the results is made publicly available to facilitate reproducibility\footnote{{https://github.com/nvdeshpa/FreqSelBeamformingBeamtrainingDMAs}}.
				
				\textbf{Universal constants:}
				Speed of light $c=3\times 10^8 ~[\text{m}/\text{sec}]$, intrinsic impedance of free space $\eta_0\approx 377~\Omega$, Boltzmann constant $k_B= 1.38 \times 10^{-23}$ J/K.
				
					\textbf{Notation}:  A bold lowercase letter $\bh$ denotes a column vector,   $(\cdot)^T$ denotes transpose,  $|\cdot|$  indicates absolute value, $\odot$ represents the element-wise product, $\mathsf{sgn}(\cdot)$ indicates the \textit{signum} function, $\bm{1}$ be the all-ones vector.

				\section{System model }\label{sec: dma preliminaries}

				\begin{figure}
					\centering 
					\includegraphics[width=0.5\textwidth]{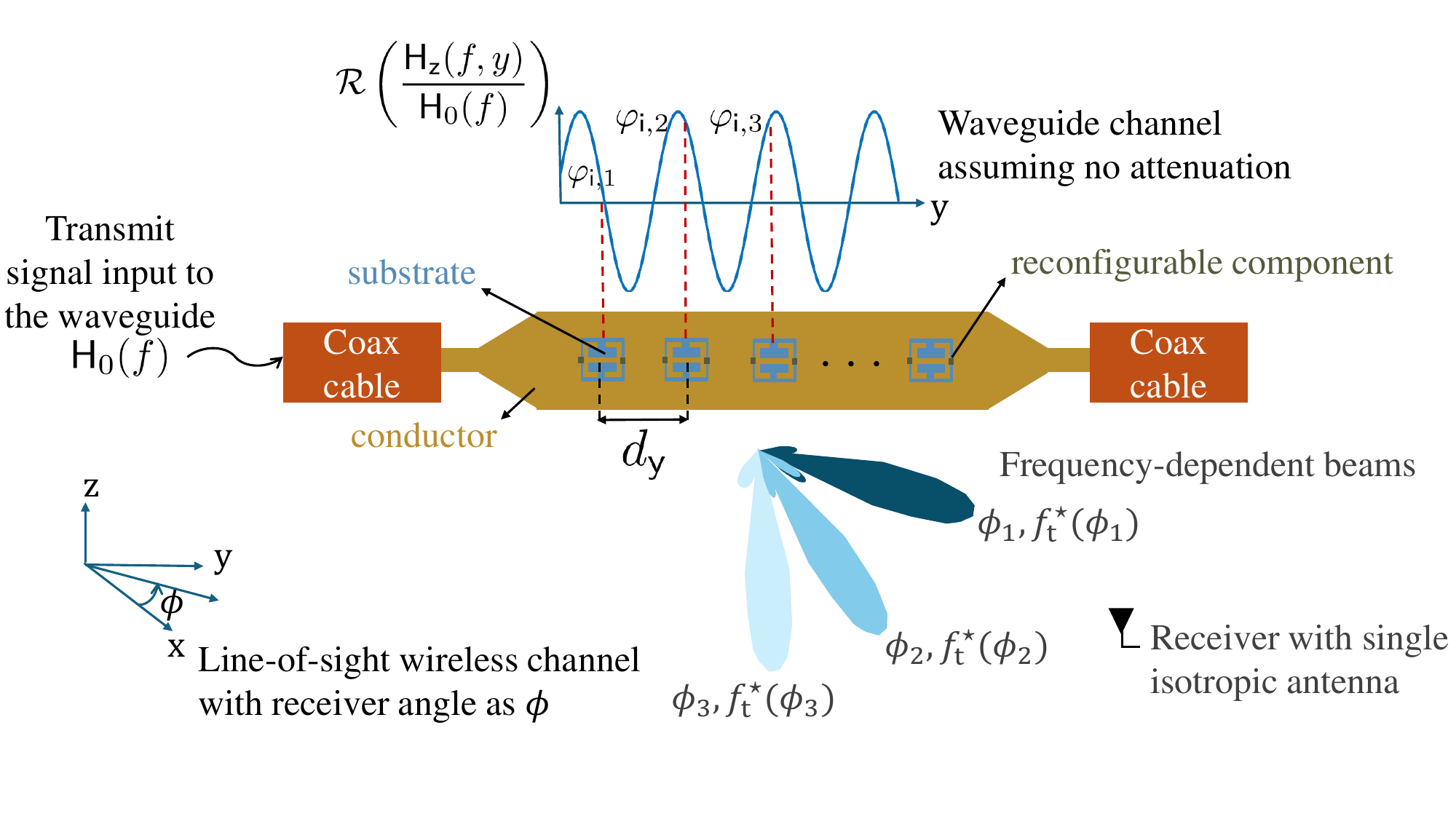}    					
					\caption{ A DMA has radiating slots  which are tuned using low-power reconfigurable components with signal incident from one end of the waveguide.
					We leverage the frequency reconfigurability property of DMA for frequency-selective beamforming. }
					\label{fig: sys model fig}
				\end{figure}
				
				\subsection{Model of a single DMA transmitter}\label{subsec: DMA as an array of reconfigurable polarizable dipoles}
				
				We describe the model of a single DMA at the transmitter.
				A DMA is a type of a  metasurface antenna that has reconfigurable radiating slots on the upper surface of a conducting waveguide.
				Let the DMA be in the $\sfy\sfz$ plane as illustrated in Fig.~\ref{fig: sys model fig}. 
					The transmit signal enters the DMA from the  waveguide input and propagates along its length. We assume the DMA has $N$ equally spaced radiating slots with inter-element spacing $d_\sfy$ along the $\sfy$ axis. 
				 %as shown in Fig.~\ref{fig: sys model fig}(a)\cite{smith_analysis_2017}. 
%				A microstrip transmission line excites an array of radiating slots from the waveguide input.  
				DMAs function as leaky-wave antennas with reconfigurable components  like varactor or PIN diodes  attached to the radiating slots to tune slot resonant frequencies.  
			 As the signal travels along the waveguide, each slot experiences a different magnitude and phase of the input signal.  The input signal radiates out of each slot, and the phase and magnitude of the transmit signal can be tuned by adjusting the slot resonant frequency.  Through this mechanism, DMAs have the capability to adjust the transmit signal without external phase-shifters or delay elements.

				The data signal is contained in the EM field entering the waveguide. Let $H_0(t)~[\text{A}/\text{m}]$ be the time domain representation of the transverse magnetic field strength at the waveguide input.  For a mathematically tractable  analysis of achievable rate, we assume that the data signal is a Gaussian wide-sense stationary random process that is completely described by its mean and second-order moment\cite{davenport1958introduction}. This signal does not have a finite energy and a windowed Fourier Transform with interval $T_0$ is used to define its spectrum\cite{russer2014modeling}.  We use  the passband frequency $f$  in Hertz for frequency-domain representation.
				Let the frequency domain magnetic field  be  $\sfH_{0}(f) = \int_{-\frac{T_0}{2}}^{\frac{T_0}{2}} H_0(t) e^{-\sfj 2\pi ft}
				\mathrm{d}t \left[\frac{\text{A}/\text{m}}{\text{Hz}}\right].$
			This frequency domain representation allows analysis of the EM fields in a communication-theoretic framework.
				
				We use the polarizable dipole framework in \cite{smith_analysis_2017}  to model the radiation characteristics of the DMAs.
				The  DMA radiating slots are much smaller compared to the free-space wavelength and can be modeled as  reconfigurable magnetic dipoles along the $\sfy$ axis.	
				% as shown in Fig.~\ref{fig: sys model fig}(b).				
				Each magnetic dipole has a frequency-selective response governed by its magnetic polarizability, which connects the induced magnetic dipole moment with the magnetic field propagating inside the waveguide.
%				 can be tuned dynamically using low-powered reconfigurable components like varactor diodes.
				Let $	\alpha_{\sfM,n}(f) \left[\text{m}^3\right]$ be the frequency response of the magnetic polarizability of the dipole corresponding to the  $n$th element on the DMA~\cite{bowen2022optimizing}.
				Unlike a conventional dipole array where each element is excited independently, each dipole is excited by the transverse magnetic field as it propagates along the DMA.
			%	The dominant component of the magnetic field inside the waveguide is the transverse component.
%				 and the beamforming is obtained by reconfiguring the  magnetic polarizability of each dipole.

%				As the reference EM signal propagates along the DMA waveguide length, there is a phase-advance of the magnetic field inside the waveguide.

 We characterize the magnetic polarizability using the Lorentzian model from \cite{smith_analysis_2017}. 
 This model captures the resonant behavior of the elements on the DMA.
  Let $\sfF \left[\text{m}^3\right]$ be the coupling factor and $\Gamma \left[\text{Hz}\right]$  be the damping factor in the Lorentzian model. 
Let $f_{\sfr, n}$ be the  resonant frequency for the $n$th dipole that is dynamically reconfigurable. 
The Lorentzian  expression for the magnetic polarizability is~\cite{smith_analysis_2017}
%Using these parameters, the magnetic polarizability $\alpha_{\sfM,n}(f)$ is expressed in the Lorentzian form as \cite{smith_analysis_2017}
\begin{align}\label{eqn: Lorentzian form}
	\alpha_{\sfM,n}(f)=  \frac{\sfF  2 \pi  f^2}{  2 \pi f_{\sfr, n}^2 - 2 \pi f^2 + \sfj \Gamma  f}\left[\text{m}^3\right].
\end{align}
We assume that $\sfF $ and $\Gamma$ are fixed and dependent on the DMA  unit cell design.

The polarizability can also be defined in terms of the quality factor  $	\sfQ=\frac{2\pi f}{\Gamma}.$
%				\begin{align}\label{eqn: Q factor}
	%				\sfQ=\frac{2\pi f}{\Gamma}.
	%				\end{align}
Let  the complex argument of the polarizability be denoted as $	\psi_{n}(f)$ and be defined as 
%			\begin{align}\label{eqn: argument v_l}
	%				\psi_{n}(f)=\begin{cases}
		%					\tan^{-1}\left(\frac{-\Gamma  f}{2\pi (f_{\sfr, n}^2- f^2)}\right), \quad \text{  for  } f_{\sfr, n} \geq  f \\
		%					\tan^{-1}\left(\frac{-\Gamma  f}{2\pi (f_{\sfr, n}^2- f^2)}\right) -\pi , \quad \text{  for  } f_{\sfr, n} < f 
		%				\end{cases}.
	%			\end{align}
\begin{align}\label{eqn: argument v_l}
	\psi_{n}(f)= \mathsf{atan2}(-\Gamma  f,  2\pi (f_{\sfr, n}^2- f^2)).
\end{align}
%We see that $	\psi_{n}(f) \in  [-\pi, 0]$.
Using 	the definition of quality factor and \eqref{eqn: argument v_l}, we express the magnetic polarizability from \eqref{eqn: Lorentzian form} as
\begin{align}\label{eqn: alpha in terms of v_lf}
	\alpha_{\sfM,n}(f)= -\sfF \sfQ \sin\left(\psi_{n}(f)\right)e^{\sfj \psi_{n}(f)} \left[\text{m}^3\right].
\end{align} 
We see that the reconfigurable resonant frequency  $f_{\sfr, n}$ modifies both the phase and magnitude of the magnetic polarizability.  The key insight from \eqref{eqn: alpha in terms of v_lf} is that the phase and magnitude of the magnetic polarizability are coupled and cannot be tuned independently.

				We model the change in the excitation signal in the DMA waveguide as it propagates by  a relative phase difference between the excitation of each polarizable dipole.
				  For a simplified analysis, 
				  we assume that there is no attenuation of the field inside the waveguide so that the
				 excitation amplitude is uniform across all DMA elements\cite{smith_analysis_2017}.
				 Including the attenuation in the model leads to a loss in the DMA beamforming gain which we discuss later in the numerical results in Section~\ref{subsec: numerical results bf gain}.
%This assumption is valid for a specific DMA length beyond which DMA attenuation should be modeled. 
				 The relative phase of the excitation  at the $n$th dipole is $e^{\sfj \varphi_{\sfi,n}(f)}$, which we term as  the instrinsic phase-shift.
				  It is dependent on the  waveguide refractive index and inter-element spacing of the DMA radiating slots, and  cannot be dynamically adjusted.
				  	Similar to \cite{smith_analysis_2017}, we assume that the transverse magnetic field component is sinusoidal in terms of the distance traveled along the $\sfy$ axis. 
				  We denote the refractive index by $n_\sfg$, where subscript $\sfg$ is used to indicate the waveguide. The $\sfz$ component of the  waveguide magnetic field with constant magnitude $|\sfH_{0}(f)|$  is expressed as $	\sfH_{\sfz}(f,y)	=	\sfH_{0}(f) e^{-\sfj n_\sfg 2 \pi \frac{f}{c} y}    \left[\frac{\text{A}/\text{m}}{\text{Hz}}\right]$.
				  %								\begin{align}
				  	%								\sfH_{z}(f,y)	=	\sfH_{0}(f) e^{-\sfj n_\sfg 2 \pi \frac{f}{c} y}    \left[\frac{\text{A}/\text{m}}{\text{Hz}}\right].
				  	%									\end{align}
				  The data signal incident on each radiating slot is governed by $	\sfH_{\sfz}(f,y)$.

								The radiating slots along the $\sfy$ axis sample  $	\sfH_{\sfz}(f,y)$ at an interval of $d_\sfy$ as shown in Fig.~\ref{fig: sys model fig}.
							The induced magnetic dipole moment $	\sfm_\sfz(f,y_{n})$ $\left[\frac{\text{A}\text{m}^2}{\text{Hz}}\right]$ along the $\sfz$ axis for the $n$th radiating slot located at distance $y_{n}$ from the waveguide input is 
							\begin{align}
								\sfm_\sfz(f,y_{n})&=	\sfH_{\sfz}(f,y_{n}) \alpha_{\sfM,n}(f) = \sfH_{0}(f)  e^{-\sfj n_\sfg 2 \pi \frac{f}{c} y_{n}}  \alpha_{\sfM,n}(f). \label{eqn:  mz in terms of H0 intrinsic phase and polarizability}
							\end{align}
							Without loss of generality, let the intrinsic phase-shift  at the first element placed at the origin be $0$.  We express the intrinsic phase-shift  at the $n$th element as 
							\begin{align}\label{eqn: intrinsic phase angle}
								\varphi_{\sfi,n}(f)= - n_\sfg 2 \pi \frac{f}{c} (n-1) d_\sfy,
							\end{align}
							and define the waveguide phase-shift vector as
							$	\bh_{\mathsf{dma}}(f)=[e^{\sfj \varphi_{\sfi,1}(f)}, \dots, e^{\sfj \varphi_{\sfi,N}(f)}]^T.$
							%				\begin{align}\label{eqn: intrinsic phase shift vector}
								%					\bh_{\mathsf{dma}}=[e^{\sfj \varphi_{\sfi,1}}, \dots, e^{\sfj \varphi_{\sfi,N}}]^T.
								%				\end{align}
							The waveguide phase-shift vector accounts for the phase-difference between the excitation at different slots.
							
				\subsection{Received signal model}\label{subsec: Propagation channel model}

				Now, we describe the  wireless propagation channel between the DMA transmitter and the single antenna receiver.
%				Besides having an intrinsic phase-shift at each element, the wireless propagation channel also contributes to a  phase-difference between each element. 
			The  phase-shift due to wireless propagation at the $n$th slot on the transmit DMA is denoted as  $e^{\sfj \varphi_{\sfee,n}(\phi,f)}$. Assuming a line-of-sight and far-field propagation where the receiver is at an azimuth angle $\phi$ from the broadside direction ($\sfx$ axis in Fig.~\ref{fig: sys model fig}), the extrinsic phase-shift angle at the $n$th element is  ~$\varphi_{\sfee,n}(\phi,f)=- 2 \pi \frac{f}{c} ~(n-~1) ~d_\sfy \sin(\phi).$ 
%				\begin{align}\label{eqn: extrinsic phase angle}
%			\varphi_{\sfee,n}= - 2 \pi \frac{f}{c} (n-1) d_\sfy \sin(\phi).
%				\end{align}			 
				Let the far-field array response  vector  be 
				$\ba(\phi, f)=[e^{\sfj \varphi_{\sfee,1}(\phi,f)}, \dots, e^{\sfj \varphi_{\sfee,N}(\phi,f)}]^T.$
%				\begin{align}\label{eqn: extrinsic phase shift vector}
%				\ba(\phi, f)=[e^{\sfj \varphi_{\sfee,1}}, \dots, e^{\sfj \varphi_{\sfee,N}}]^T.
%				\end{align}
				Note that the extrinsic phase-shift angle 	$\varphi_{\sfee,n}(\phi,f)$ depends on the transmit angle of departure (AoD) $\phi$ which can vary  unlike the intrinsic phase-shift angle $\varphi_{\sfi,n}(f)$ that depends only on the fixed DMA waveguide design. 
%				The extrinsic phase-shift  vector is also known as the array steering vector for conventional antenna arrays. 
%					For DMAs, we distinguish between the extrinsic phase-shift and intrinsic  phase-shift. 
%					The overall phase difference between different elements on the DMA is superposition of intrinsic and extrinsic phase differences. 
				 
				 The effective propagation channel is expressed as
					$	\bh(\phi,f)=~\bh_{\mathsf{dma}}(f) \odot 	\ba(\phi, f).$
%				\begin{align}\label{eqn: element wise prod of extrinsic and intrinsic}
%					\bh(\phi,f)= 	\bh_{\mathsf{dma}} \odot 	\ba(\phi, f).
%				\end{align}
				The phase  of the $n$th element in $\bh(\phi,f)$ is 
				\begin{align}\label{eqn: varphi ell def}
					\varphi_{n}(\phi,f)&=\varphi_{\sfi,n}(f)+\varphi_{\sfee,n}(\phi,f)\\ &=- 2 \pi \frac{f}{c} (n-1) d_\sfy (n_\sfg +\sin(\phi)).
				\end{align}
				For a conventional array, $\varphi_{\sfi,n}(f)=0$. Hence,  $\varphi_{n}(\phi,f)=\varphi_{\sfee,n}(\phi,f)$. For a matched filter strategy,  this implies  applying an equal and opposite phase-shift $\varphi_{\sfr,n}(\phi,f)=-\varphi_{\sfee,n}(\phi,f)$ to beamform toward  $\phi$.
				For a DMA,   the transmit beamforming should account for the additional term $\varphi_{\sfi,n}(f)$ in the overall phase.

%				The radiation through each slot is modeled as a radiation of a magnetic dipole.
We define the received power spectral density in terms of the transmit power spectral density.
				 Assuming that mutual coupling is negligible between each radiating slot, the far-field component of the electric field at the receiver in the $\sfx\sfy$ plane at a distance $r$ from the DMA is~\cite{smith_analysis_2017}
%				\begin{align}
%					\sfE_{\phi}(f)&=\eta_0\left(\frac{2 \pi  f}{c}\right)^2 \sum_{n=1}^N  	\sfm_z(f,y_{n}) \frac{e^{\sfj \varphi_{\sfee,n}}}{4\pi r}  \left[\frac{\text{V}/\text{m}}{\text{Hz}}\right], \nonumber\\
%					&\stackrel{(a)}{=}\frac{\eta_0}{4\pi r}\left(\frac{2 \pi  f}{c}\right)^2  \sfH_{0}(f)  \sum_{n=1}^N  	  \alpha_{\sfM,n}(f) {e^{\sfj \varphi_{n}}} \left[\frac{\text{V}/\text{m}}{\text{Hz}}\right], \label{eqn: E phi in terms of polarizability}\nonumber\\
%					&\stackrel{(b)}{=} \frac{\eta_0}{4\pi r}\left(\frac{2 \pi  f}{c}\right)^2  \sfH_{0}(f) \sfF \sfQ   \bff^T_{\mathsf{DMA}}(f) \bh(\phi,f) \left[\frac{\text{V}/\text{m}}{\text{Hz}}\right].
%				\end{align}
				\begin{align}
					\sfE_{\phi}(f)&=\eta_0\left(\frac{2 \pi  f}{c}\right)^2 \sum_{n=1}^N  	\sfm_z(f,y_{n}) \frac{e^{\sfj \varphi_{\sfee,n}(\phi,f)}}{4\pi r}  \left[\frac{\text{V}/\text{m}}{\text{Hz}}\right], \nonumber\\
					&\stackrel{(a)}{=}\frac{\eta_0}{4\pi r}\left(\frac{2 \pi  f}{c}\right)^2  \sfH_{0}(f)  \sum_{n=1}^N  	  \alpha_{\sfM,n}(f) {e^{\sfj \varphi_{n}(\phi,f)}} \left[\frac{\text{V}/\text{m}}{\text{Hz}}\right]. \label{eqn: E phi in terms of polarizability}
				\end{align}
				Here, equality $(a)$ follows from \eqref{eqn:  mz in terms of H0 intrinsic phase and polarizability} and  \eqref{eqn: varphi ell def} by substituting the expression of the induced magnetic dipole moment.
				
			The receive electric field shows how the polarizability terms beamform the signal from each slot.
				To make the analysis consistent with traditional beamforming, 	we define the dimensionless DMA beamforming vector as 
				\begin{align}\label{eqn: f dma}
					\bff_{\mathsf{DMA}}(f)=\left[-\sin\left(\psi_{1}(f)\right)e^{\sfj \psi_{1}(f)}, \dots, -\sin\left(\psi_{N}(f)\right)e^{\sfj \psi_{N}(f)} \right]^T .
				\end{align}
		This yields
				\begin{align}\label{eqn: E phi f closed form}
			\sfE_{\phi}(f)	=	\frac{\eta_0}{4\pi r}\left(\frac{2 \pi  f}{c}\right)^2  \sfH_{0}(f) \sfF \sfQ   \bff^T_{\mathsf{DMA}}(f) \bh(\phi,f) \left[\frac{\text{V}/\text{m}}{\text{Hz}}\right].
				\end{align}
%				and equality $(b)$ follows from  \eqref{eqn: element wise prod of extrinsic and intrinsic}, \eqref{eqn: alpha in terms of v_lf} and \eqref{eqn: f dma}.
%				From \eqref{eqn: E phi in terms of polarizability}, we see that the magnetic polarizability $\alpha_{\sfM,n}(f)$ can be treated as a beamforming weight for the $n$th element on the DMA. We define the DMA beamforming vector as
%				\begin{align}\label{eqn: fdma}
%					\bff_{\mathsf{DMA}}(f)=[\alpha_{\sfM,1}(f), \dots, \alpha_{\sfM,N}(f)]^T \left[\text{m}^3\right].
%				\end{align}
%				Using \eqref{eqn: fdma}, we express the electric field far-field component using the DMA beamforming vector and effective propogation channel vector as 
%				\begin{align}
%				\sfE_{\phi}(f)&= \frac{-\pi \eta_0 f^2}{r c^2} |\sfH_{0}(f) |\bff^T_{\mathsf{DMA}}(f)	\bh(\phi,f)\left[\frac{\text{V}/\text{m}}{\text{Hz}}\right].
%				\end{align}
				Similarly, the far-field magnetic field strength component perpendicular to the $\sfx\sfy$ plane is $	\sfH_{\theta}(f)= -\frac{1}{\eta_0} \sfE_{\phi}(f) \left[\frac{\text{A}/\text{m}}{\text{Hz}}\right].$
%				\begin{align}
%					\sfH_{\theta}(f)= -\frac{1}{\eta_0} \sfE_{\phi}(f) \left[\frac{\text{A}/\text{m}}{\text{Hz}}\right].
%				\end{align}
				The radial component of the Poynting vector is 
				\begin{align}
					&\sfS_{\sfr}(f)=-\frac{1}{2}\sfE_{\phi}(f)\sfH_{\theta}^c(f) =\frac{\eta_0}{2}|\sfH_{\theta}(f)|^2 \left[\frac{\text{W}/\text{m}^2}{\text{Hz}^2}\right] ,\\
%					&= \frac{\eta_0}{2} \left(\frac{1}{4\pi r}\right)^2  \left(\frac{2 \pi  f}{c}\right)^4 |\sfH_{0}(f)|^2 \sfF^2 \sfQ^2   |\bff^T_{\mathsf{DMA}}(f)\bh(\phi,f)|^2  \\
					&= \!\!\frac{1}{4\pi r^2}  \frac{\eta_0}{8 \pi}\! \left( \!\!\frac{4 \pi^2  f^2 |\sfH_{0}(f)| \sfF \sfQ \left|\bff^T_{\mathsf{DMA}}(f) \bh(\phi,f)\right| }{c^2} \!\! \right)^2   \left[\frac{\text{W}}{\text{m}^2\text{Hz}^2}\right]\!.
%					&=  \frac{\eta_0}{2} \left(\frac{\pi  f^2}{r c^2}\right)^2 |\sfH_{0}(f)|^2 \sfF^2 \sfQ^2   |\bff^T_{\mathsf{DMA}}(f) \bh(\phi,f)|^2 \left[\frac{\text{W}/\text{m}^2}{\text{Hz}^2}\right] 
				\end{align}
				Let the transmit power spectral density be defined as $	P_{\mathsf{T}}(f) \!=\! \underset{T_0 \rightarrow \infty}{\mbox{lim}}\frac{1}{T_0} \frac{\eta_0}{8 \pi} \left( \frac{4 \pi^2  f^2}{c^2} \sfF \sfQ  \right)^2\!\! \bbE[|\sfH_{0}(f) |^2]  \left[\frac{\text{W}}{\text{Hz}}\right]\!$
				 and the frequency-selective beamforming gain term be defined as 
				 \begin{align}\label{eqn: G dma}
				 	\cG_{\mathsf{DMA}}(\phi, f)=  \left|\bff^T_{\mathsf{DMA}}(f) \bh(\phi,f)\right|^2 .
				 \end{align}
%				\begin{align}
%						P_{\mathsf{T}}(f) \!=\! \underset{T_0 \rightarrow \infty}{\mbox{lim}}\frac{1}{T_0} \frac{\eta_0}{8 \pi} \left( \frac{4 \pi^2  f^2}{c^2} \sfF \sfQ  \right)^2\!\! \bbE[|\sfH_{0}(f) |^2]  \left[\frac{\text{W}}{\text{Hz}}\right]\!.
%				\end{align}
%			  At wavelength $\lambda$, the effective isotropic antenna aperture is $\frac{\lambda^2}{4 \pi}$.
			The received power spectral density is  expressed in terms of the transmit power spectral density as
			\begin{align}
				&P_{\mathsf{R}}(\phi,f) =  \underset{T_0 \rightarrow \infty}{\mbox{lim}}\frac{1}{T_0} \frac{\lambda^2}{4 \pi} \bbE[\sfS_{\sfr}(f)] \\ &=\left(\frac{\lambda}{4 \pi r}\right)^2 P_{\mathsf{T}}(f) 	\cG_{\mathsf{DMA}}(\phi, f) \left[\frac{\text{W}}{\text{Hz}}\right].
			\end{align}
				The DMA configuration affects the receive signal power through the frequency-selective beamforming gain.
%				\begin{align}
%					P_{\mathsf{R}}(\phi,f) = \left(\frac{\lambda}{4 \pi r}\right)^2 P_{\mathsf{T}}(f)  \left|\bff^T_{\mathsf{DMA}}(f) \bh(\phi,f)\right|^2  \left[\frac{\text{W}}{\text{Hz}}\right].
%				\end{align}

	In this paper, we compare the DMA-based transmitter with a TTD architecture~\cite{10051948}.
%For this comparison, we also define the TTD beamforming vector.
%The frequency-flat beamforming vector for a conventional phased-array from Fig.~\ref{fig: sys model fig}(c) is $\bff_{\mathsf{PS}}=\left[e^{\sfj \varphi_{\sfr, 1}}, \dots, e^{\sfj \varphi_{\sfr,N}} \right]^T $
%%					\begin{align}\label{eqn: f ps}
%	%						\bff_{\mathsf{PS}}=\left[e^{\sfj \varphi_{\sfr, 1}}, \dots, e^{\sfj \varphi_{\sfr,N}} \right]^T ,
%	%					\end{align}
%and the corresponding beamforming gain is $	\cG_{\mathsf{PS}}(\phi,f)=\left|\bff^T_{\mathsf{PS}} \bh_{\sfee}(\phi,f)\right|^2$.
The TTD beamforming architecture for a conventional array  with $\tau_n$ delay at $n$th element produces a frequency-selective beamforming vector as 
$\bff_{\mathsf{TTD}}(f)=\left[e^{\sfj 2\pi f\tau_1}, \dots, e^{\sfj 2\pi f \tau_N} \right]^T $
%					\begin{align}\label{eqn: f ttd}
	%						\bff_{\mathsf{TTD}}(f)=\left[e^{\sfj 2\pi f\tau_1}, \dots, e^{\sfj 2\pi f \tau_N} \right]^T ,
	%					\end{align}
and the corresponding beamforming gain is 
$	\cG_{\mathsf{TTD}}(\phi,f)=\left|\bff^T_{\mathsf{TTD}}(f) \bh(\phi,f)\right|^2$ where the intrinsic phase vector is an all ones vector.
%Note the difference that $\cG_{\mathsf{TTD}}(\phi,f)$ depends on $\bh_{\sfee}(\phi,f)$ whereas $\cG_{\mathsf{DMA}}(\phi, f)$ depends on $\bh(\phi,f)$.  This is because there is no instrinsic phase-difference in conventional antenna arrays.
Prior work has compared the frequency-selective beamforming gain of the TTD and phased-array but comparison with DMAs is missing~\cite{10002944}. 
We choose TTD as a benchmark for comparison.
DMA beamforming is energy-efficient but is complicated by the coupled phase and magnitude constraint.
%						as shown in \eqref{eqn: f dma}
%whereas TTD  beamforming weights have unit magnitude. 
%						as shown in \eqref{eqn: f ps} and \eqref{eqn: f ttd}.
%					Comparing the beamforming structure of a DMA in \eqref{eqn: f dma}, phased-array in \eqref{eqn: f ps}, and TTD in \eqref{eqn: f ttd}, we see that the DMA beamforming is the most challenging because of the coupled magnitude and phase. 
%					We address this challenge of coupled magnitude and phase constraint in Section~\ref{sec: Frequency-selective beamforming gain optimization and analysis}.			
%					For comparison purposes, we also similarly define the beamforming gain for a phased-array as $	\cG_{\mathsf{PS}}(\phi,f)=\left|\bff^T_{\mathsf{PS}} \bh_{\sfee}(\phi,f)\right|^2$ and for a TTD array as  $	\cG_{\mathsf{TTD}}(\phi,f)=\left|\bff^T_{\mathsf{TTD}}(f) \bh_{\sfee}(\phi,f)\right|^2$.
In Section~\ref{sec: Frequency-selective beamforming gain optimization and analysis}, we discuss how to design the DMA beamformer despite these limitations.

	\subsection{Signal to noise ratio and achievable rate}\label{sec: snr and bf gain}

We measure the system performance  in terms of the achievable rate.
					Let  $k_B$ be  the Boltzmann constant in J/K and  $T$ be the noise temperature of the antenna in K. The noise spectral density at the receiver is $	 \mathsf{N}_0= k_B T \left[\frac{\text{W}}{\text{Hz}}\right]$. The received SNR for a line-of-sight receiver at angle $\phi$ operating at frequency $f$ is  defined as $	\mathsf{SNR}(\phi,f)= \frac{P_{\mathsf{R}}(\phi,f)}{\mathsf{N}_0}.$
%					\begin{align}\label{eqn: snr}
%							\mathsf{SNR}(\phi,f)= \frac{P_{\mathsf{R}}(\phi,f)}{\mathsf{N}_0}.
%					\end{align}
						The received SNR varies with frequency.

					We model a communication system where the center frequency of the data transmission band can be dynamically reconfigured	as the DMA has the ability to resonate at different frequencies.
						Let	the communication bandwidth be fixed as $\sfB$.  Let the operating frequency allocated to the user by the transmitter be $f_{\sft}$   which can be dynamically chosen from the range  $[f_{\sft,\mathsf{min}}, f_{\sft,\mathsf{max}}]$.
						The DMA transmitter transfers data in the band $[f_{\sft}- \frac{\sfB}{2}, f_{\sft}+\frac{\sfB}{2}]$.
						In the data transmission band, we assume an OFDM system with $K_{\sfd}$ subcarriers. 				
						The achievable rate for a receiver with AoD as $\phi$ and  operating frequency $f_{\sft}$ is defined as
% $\mathsf{R}(\phi, f_\sft)=\frac{\sfB}{K_{\sfd}}\sum_{k=1}^{K_{\sfd} }\log_2\left(1+\mathsf{SNR}(\phi,f_k)\right) \left[\mathsf{bits/s}\right]$.  
						\begin{align}
							\mathsf{R}(\phi, f_\sft)=\frac{\sfB}{K_{\sfd}}\sum_{k=1}^{K_{\sfd} }\log_2\left(1+\mathsf{SNR}(\phi,f_k)\right) \left[\text{bits/s}\right].
						\end{align}
							In this work, we focus on the problem of configuring the resonant frequencies of the transmit DMA and choosing the  operating frequency to maximize the frequency-selective beamforming gain at the receiver. Optimizing the beamforming gain will result in higher received SNR.
							Assuming equal power allocation across frequency, maximizing SNR is equivalent to maximizing achievable rate.

				\section{Frequency-selective beamforming gain optimization and analysis}\label{sec: Frequency-selective beamforming gain optimization and analysis}
				
				The main focus of this section is to find the DMA configuration that achieves the highest beamforming gain over a line-of-sight channel. The  AoD is assumed to be perfectly known in this section. We propose a method for AoD estimation in Section~\ref{sec: Single-shot beam training with DMAs}.
					We propose a two-stage optimization approach to maximize the gain at a specific  AoD.  In the first stage, we develop closed-form expressions for the DMA resonant frequencies to maximize the beamforming gain at a particular operating frequency and desired  angle.  Using the results of Section~\ref{ref: Beamforming gain optimization for a target frequency and angle }, we propose a second-stage optimization over the operating frequency in Section~\ref{subsec: Optimizing target frequency for DMA beamforming gain}. The prior work on DMA beamforming is limited only to the first stage optimization~\cite{carlson2023hierarchical, 9991243}, which we show is sub-optimal compared to our two-stage optimization approach in the proposed wideband setting.
					Although the first stage optimization has been studied, our formulation differs from the prior work as we explicitly solve for the resonant frequency.

		\subsection{Beamforming gain optimization}\label{ref: Beamforming gain optimization for a target frequency and angle }
	
%	In Section~\ref{sec: snr and bf gain}, we defined the frequency-selective beamforming gain $\cG_{\mathsf{DMA}}(\phi, f)$ in terms of the DMA precoder $\bff_{\mathsf{DMA}}(f)$ and the effective propagation channel $\bh(\phi,f)$.
	 We formulate the problem of maximizing the beamforming gain for a particular  angle $\phi_{\mathsf{AoD}}$ and operating frequency $f_{\sft}$ in terms of the reconfigurable resonant frequency of each DMA element. We assume that the resonant frequencies  can attain any positive real value $\{f_{\sfr, n} \}_{n=1}^N$.  We mathematically formulate the problem as follows:
	\begin{subequations}\label{problem1}
		\begin{alignat}{3}
			\textbf{P1}:& \underset{ \{f_{\sfr, n} \}_{n=1}^N}{\mbox{ max }} \cG_{\mathsf{DMA}}(\phi_{\mathsf{AoD}}, f_{\sft})
			, \\
			&\text{ s.t. } \{f_{\sfr, n} \}_{n=1}^N>0 .
			 \label{problem1_a}
		\end{alignat}
	\end{subequations}
To highlight why problem \textbf{P1} is challenging compared to the conventional TTD beamforming and to use some of the insights from the TTD solution to address the challenge in DMA beamforming,  we present problem \textbf{P2} for maximizing $\cG_{\mathsf{TTD}}(\phi_{\mathsf{AoD}}, f_{\sft})$.
	For a TTD array, the mathematical formulation for maximizing the beamforming gain at   angle $\phi_{\mathsf{AoD}}$ and operating frequency $f_{\sft}$ is 
				\begin{subequations}\label{problem3}
					\begin{alignat}{3}
						\textbf{P2}:& \underset{ \{ \tau_{n} \}_{n=1}^N}{\mbox{ max }} \cG_{\mathsf{TTD}}(\phi_{\mathsf{AoD}}, f_{\sft})
						, \\
						&\text{ s.t. } \{\tau_{n} \}_{n=1}^N \geq 0.
						\label{problem3_a}
					\end{alignat}
				\end{subequations}		
				The optimal solution for \textbf{P2} is  expressed in terms of  $\phi_{\mathsf{AoD}}$ as 
				\begin{align}\label{eqn: ttd closed form soln}
					\tau^{\star}_{n}= \begin{cases}
						\frac{d_\sfy}{c} (n-1)  \sin(\phi_{\mathsf{AoD}}), \quad \phi_{\mathsf{AoD}}\geq 0,\\
							\frac{d_\sfy}{c} (n -N)  \sin(\phi_{\mathsf{AoD}})				, \quad \phi_{\mathsf{AoD}}<0.		
					\end{cases}
				\end{align}
				The maximum value of the beamforming gain is $\cG^{\star}_{\mathsf{TTD}}(\phi_{\mathsf{AoD}}, f_{\sft})=N^2$.
				
%				At a single frequency, TTD and phased-array beamforming have the same maximum gain. 
				For wideband transmission, the beam squint-effect leads to reduction in beamforming gain for frequencies away from the center frequency. This degradation in phased-array gain has been studied in comparison to TTD in prior work~\cite{10002944} but comparison with  a DMA is missing. The DMA has higher frequency-selectivity compared to phased-array. 
%				In Section~\ref{subsec: DMA beamforming gain frequency response and half-power bandwidth analysis}, we bridge this gap by analyzing the DMA frequency response. For Section~\ref{ref: Beamforming gain optimization for a target frequency and angle } and Section~\ref{subsec: Optimizing target frequency for DMA beamforming gain}, we focus on optimizing the beamforming gain at a single frequency only. 
With this solution in \eqref{eqn: ttd closed form soln}, the beam squint is eliminated completely for a line-of-sight channel, whereas for DMA this is not possible. 
				The solution for \textbf{P2} is straightforward  because the beamforming weight lies on a unit modulus circle, i.e., the magnitude for the beamforming weight is fixed. 
%				 as shown in Fig.~\ref{fig: bf gain complex response}. 
				  We use this observation to reformulate the objective in \textbf{P1}.

					 The optimization variable $f_{\sfr, n}$ in \textbf{P1} appears in both the phase and magnitude of the DMA beamforming weight, which makes it challenging to solve the problem directly in terms of  the resonant frequencies $f_{\sfr, n}$. In \eqref{eqn: f dma}, we observe that the magnitude of $	[\bff_{\mathsf{DMA}}(f)]_{n} $ is sinusoidal in terms of the phase angle  $	\psi_{n}(f) \in  [-\pi, 0]$. 
%					From \eqref{eqn: argument v_l}, we see that $	\psi_{n}(f) \in  [-\pi, 0]$.
					 In the complex domain,  the locus of the beamforming weight $	[\bff_{\mathsf{DMA}}(f)]_{n} $ is a circle with  radius of $0.5$  centered at $\frac{-\sfj}{2}$.
%					  as shown in Fig.~\ref{fig: bf gain complex response}.
					 Compared to a phase-shifter or TTD weight with $\varphi_{\sfr, n} \in [0, 2\pi)$, the DMA beamforming weight has  a limited phase-range in   $	\psi_{n}(f) \in  [-\pi, 0]$. To simplify the DMA beamforming constraint, we define a phase angle $\widetilde{\psi}_{n}(f)$ which covers the whole angular range of $2\pi$.
					 We use a shift-of-origin transformation for the DMA beamforming weight as \cite{carlson2023hierarchical}
					 \begin{align}\label{eqn: shift of origin}
					 [\bff_{\mathsf{DMA}}(f)]_{n}= \frac{-\sfj+e^{\sfj \widetilde{\psi}_{n}(f)}}{2}.
					 \end{align}
					  The relation between $\psi_{n}(f)$ and $\widetilde{\psi}_{n}(f)$ is 
					 	\begin{equation}\label{eqn: v theta in terms of v tilde }
					 		\psi_{n}(f)= \frac{\widetilde{\psi}_{n}(f)}{2} -\frac{\pi}{4},
					 	\end{equation}
					 	where $\widetilde{\psi}_{n}(f) \in \left[\frac{-3\pi}{2}, \frac{\pi}{2}\right]$.
				In problem $\textbf{P1}$, the objective is to maximize the beamforming gain at $f_{\sft}$ by optimizing the resonant frequencies $\{f_{\sfr, n} \}_{n=1}^N$.  Substituting $f=f_{\sft}$ in \eqref{eqn: argument v_l} and \eqref{eqn: v theta in terms of v tilde }, we express the resonant frequency of the $n$th slot at operating frequency $f_{\sft}$ as
%				\begin{align}\label{eqn: frl in terms of v tilde}
%					f_{\sfr, n}=\sqrt{f_{\sft}^2-\frac{\Gamma f_{\sft}}{2\pi \tan\left(\frac{\widetilde{\psi}_{n}(f_{\sft})}{2} -\frac{\pi}{4}\right)}}.
%				\end{align}
				\begin{align}\label{eqn: frl in terms of v tilde}
					f_{\sfr, n}=\sqrt{f_{\sft}^2+\frac{\Gamma f_{\sft}}{2\pi }\tan\left( \frac{\pi}{4} +\frac{\widetilde{\psi}_{n}(f_{\sft})}{2}\right)}.
				\end{align}
				The shift-of-origin transformation simplifies the beamforming gain optimization in $\textbf{P1}$  by
			 reformulating the problem in terms of $\widetilde{\psi}_{n}(f)$ and then computing the optimal $f_{\sfr, n}$ using 
				\eqref{eqn: frl in terms of v tilde}.
				
				A shift-of-origin transformation can also be applied to the DMA beamformer to
			 define the unit-modulus DMA beamforming vector  $\widetilde{\bff}_{\mathsf{DMA}}(f)$ such that $[\widetilde{\bff}_{\mathsf{DMA}}(f)]_{n}= e^{\sfj \widetilde{\psi}_{n}(f)}$. The original DMA beamforming vector is expressed in terms of $\widetilde{\bff}_{\mathsf{DMA}}(f)$  as $\bff_{\mathsf{DMA}}(f)= \frac{-\sfj \bm{1}+ \widetilde{\bff}_{\mathsf{DMA}}(f)}{2}.$
%				\begin{align}\label{eqn: f dma in terms of f tilde dma}
%				\bff_{\mathsf{DMA}}(f)= \frac{-\sfj \bm{1}+ \widetilde{\bff}_{\mathsf{DMA}}(f)}{2}.
%				\end{align}
			The DMA beamforming gain in \eqref{eqn: G dma} is
				$\cG_{\mathsf{DMA}}(\phi_{\mathsf{AoD}}, f)=	
				\frac{1 }{4} \left|-\sfj \bm{1}^T \bh(\phi_{\mathsf{AoD}},f) + \widetilde{\bff}^T_{\mathsf{DMA}}(f) \bh(\phi_{\mathsf{AoD}},f)\right|^2 $.
We reformulate  \textbf{P1} in terms of the unit-modulus beamforming weight as 

\begin{subequations}\label{problem1a}
	\begin{alignat}{3}
		\textbf{P1a}:& \underset{ \{\widetilde{\psi}_{n}(f_{\sft}) \}_{n=1}^N}{\mbox{ max }} \cG_{\mathsf{DMA}}(\phi_{\mathsf{AoD}}, f_{\sft})
		, \label{problem1a_a}\\
		&\text{ s.t. }\{\widetilde{\psi}_{n}(f_{\sft}) \}_{n=1}^N \in  \left[\frac{-3\pi}{2}, \frac{\pi}{2}\right].
		\label{problem1a_b}
	\end{alignat}
\end{subequations}
The closed-form expression for the optimal $\widetilde{\psi}_{n}(f_{\sft})$ and the expression for the maximum value of beamforming gain at $\phi_{\mathsf{AoD}}$ and $f_{\sft}$ are given in the following lemma. 
%For the closed-form expressions, we define the term $	\cS(\phi_{\mathsf{AoD}}, f_{\sft})$ as
%\begin{align}
%	\cS(\phi_{\mathsf{AoD}}, f_{\sft})= \frac{\sin\left( \pi N\frac{ f_{\sft}d_\sfy (n_\sfg +\sin(\phi_{\mathsf{AoD}}))}{c} \right)}{\sin\left( \pi \frac{ f_{\sft}d_\sfy (n_\sfg +\sin(\phi_{\mathsf{AoD}}))}{c} \right)}.
%\end{align}

\begin{lemma}\label{thm:  optimal soln P1a} Let $	\cS(\phi, f_{\sft})= \frac{\sin\left( \pi N\frac{ f_{\sft}d_\sfy (n_\sfg +\sin(\phi))}{c} \right)}{\sin\left( \pi \frac{ f_{\sft}d_\sfy (n_\sfg +\sin(\phi))}{c} \right)}$.
		The optimal solution to problem \textbf{P1a}  $\{\widetilde{\psi}^{\star}_{n}(f_{\sft}) \}_{n=1}^N$ takes the form
	\begin{align}\label{eqn: vell opt}
		&\widetilde{\psi}^{\star}_{n}(f_{\sft}) =-\frac{\pi}{2}\mathsf{sgn}(\cS(\phi_{\mathsf{AoD}}, f_{\sft})) \nonumber \\&+ \frac{2\pi f_{\sft}d_\sfy (n_\sfg +\sin(\phi_{\mathsf{AoD}}))}{c} \left(n-1- \frac{N-1}{2}\right) +2m\pi,
	\end{align}
	where $m$ is an integer such that $\widetilde{\psi}^{\star}_{n}(f_{\sft}) \in \left[\frac{-3\pi}{2}, \frac{\pi}{2}\right] $.  The maximum value of the beamforming gain at $\phi_{\mathsf{AoD}}$ and $f_{\sft}$ is 
	\begin{align}\label{eqn: Gstar DMA phi t ft}
		\cG^{\star}_{\mathsf{DMA}}(\phi_{\mathsf{AoD}}, f_{\sft})= \frac{1}{4}\left(N+ |\cS(\phi_{\mathsf{AoD}}, f_{\sft})|\right)^2.
	\end{align}
\end{lemma}
\begin{proof}
%The proof is in Appendix~\ref{proof: lemma optimal soln P1a}.
We bound the		beamforming gain from \eqref{problem1a_a} as
\begin{align}
	&\cG_{\mathsf{DMA}}(\phi, f_{\sft})
	%		=  0.25 \bigg(\left|-\sfj \bm{1}^T \bh(\phi, f_{\sft}) \right|^2 + \left|\widetilde{\bff}^T_{\mathsf{DMA}}(f_{\sft}) \bh(\phi, f_{\sft})\right|^2  \nonumber\\  &+ 2 \cR\{ (-\sfj \bm{1}^T \bh(\phi, f_{\sft}))^c (\widetilde{\bff}^T_{\mathsf{DMA}}(f_{\sft}) \bh(\phi, f_{\sft})) \} \bigg) ,  \\
	\stackrel{(a)}{\leq}  0.25 \bigg(\bigg|-\sfj \bm{1}^T \bh(\phi, f_{\sft}) \bigg|^2 + \bigg|\widetilde{\bff}^T_{\mathsf{DMA}}(f_{\sft}) \bh(\phi, f_{\sft})\bigg|^2 \nonumber\\  & + 2  \bigg|-\sfj \bm{1}^T \bh(\phi, f_{\sft}) \bigg| \bigg|\widetilde{\bff}^T_{\mathsf{DMA}}(f_{\sft}) \bh(\phi, f_{\sft})\bigg| \bigg).
\end{align}
We simplify the term $-\sfj \bm{1}^T \bh(\phi, f_{\sft})= \cS(\phi, f_{\sft})  e^{\sfj \left( -\frac{\pi}{2} -\pi (N-1)\frac{ f_{\sft}d_\sfy (n_\sfg +\sin(\phi))}{c} \right)}$.
Equating the phase of the terms  	$-\sfj \bm{1}^T \bh(\phi, f_{\sft})$ and $\widetilde{\bff}^T_{\mathsf{DMA}}(f_{\sft}) \bh(\phi, f_{\sft})$, we obtain
equality in $(a)$ which proves the desired result. 
\end{proof}

From the  expression in \eqref{eqn: vell opt}, we observe that the optimal phase can be decomposed into two terms \textemdash~ a  term dependent on the slot index $n$ and a term independent of  $n$. The term independent of $n$ is the constant phase rotation applied to each element to maximize the gain. 
The beamforming gain is not invariant to the common phase rotation because of the Lorentzian constraint.
 Our proposed solution resembles the optimal phase rotation approach discussed in the prior work\cite{carlson2023hierarchical}.  We  propose a closed-form solution that is applicable for a line-of-sight channel and for the ideal case with no waveguide attenuation unlike \cite{carlson2023hierarchical} which used a brute-force search approach to compute the optimal phase rotation in a setting with a non-ideal DMA.

 The optimal resonant frequencies are computed by substituting the closed-form expression for $\widetilde{\psi}^{\star}_{n}(f_{\sft}) $ from \eqref{eqn: vell opt} in \eqref{eqn: frl in terms of v tilde}.
 Let	the optimal solution to problem \textbf{P1} be  $\{f_{\sfr, n}^{\star} \}_{n=1}^N$.
 We express it in closed-form as 
\begin{align}\label{eqn: f star l closed form}
	f_{\sfr, n}^{\star}&= \bigg(   f_{\sft}^2+\frac{\Gamma f_{\sft}}{2\pi }\tan\bigg( \frac{\pi}{4}(1-\mathsf{sgn}(\cS(\phi_{\mathsf{AoD}}, f_{\sft})))  \nonumber \\ &+\frac{\pi f_{\sft}d_\sfy (n_\sfg +\sin(\phi_{\mathsf{AoD}}))}{c} \bigg(n-1- \frac{N-1}{2}\bigg)\bigg)
	\bigg)^{\frac{1}{2}}.
\end{align}
 	We now have a mapping between the DMA resonant frequencies to the receiver  angle as a parametric function of the operating frequency.  For a fixed $f_{\sft}$, \eqref{eqn: f star l closed form} is the optimal solution to configure the DMA.  In Section~\ref{subsec: Optimizing target frequency for DMA beamforming gain}, we go one step further and optimize the communication system with a flexible choice of  operating frequency that can leverage the DMA frequency reconfigurability to further improve the beamforming gain.

\subsection{Optimizing operating frequency and  angular coverage}\label{subsec: Optimizing target frequency for DMA beamforming gain}

The maximum value of $	\cG^{\star}_{\mathsf{DMA}}(\phi_{\mathsf{AoD}}, f_{\sft})$ varies with $\phi_{\mathsf{AoD}}$ and $ f_{\sft}$ as proved in Lemma~\ref{thm:  optimal soln P1a}.  For a TTD array, the maximum beamforming gain is independent of $\phi_{\mathsf{AoD}}$ and $ f_{\sft}$ and equals $N^2$. This begs the question whether it is possible to attain the maximum value of $N^2$ for  a  DMA too?
For a DMA operating in a  system with wider spectrum availability,  we formulate a problem of optimally choosing the operating frequency to maximize the beamforming gain at a particular  angle.  We assume that the 
operating frequency range of the DMA is $[f_{\sft,\mathsf{min}}, f_{\sft, \mathsf{max}}]$.
Mathematically,  \textbf{P1} can be modified to include the operating frequency optimization as follows:
\begin{subequations}\label{problem4}
	\begin{alignat}{3}
		\textbf{P3}:& \underset{ f_{\sft}}{\mbox{ max }} \left(\underset{ \{f_{\sfr, n} \}_{n=1}^N}{\mbox{ max }} \cG_{\mathsf{DMA}}(\phi_{\mathsf{AoD}}, f_{\sft})\right)
		, \\
		&\text{ s.t. } \{f_{\sfr, n} \}_{n=1}^N>0 ,
		\label{problem4_a}\\
		& f_{\sft} \in [f_{\sft,\mathsf{min}}, f_{\sft,\mathsf{max}}] \label{eq: ft range}.
	\end{alignat}
\end{subequations}
In the following lemma, we  prove that it is possible to attain the same maximum value as that of a TTD array by optimizing over the operating frequency.

\begin{lemma}\label{thm:  optimal target freq}
	 Let $p_{\mathsf{min}}= \frac{ f_{\sft,\mathsf{min}}d_\sfy (n_\sfg +\sin(\phi_{\mathsf{AoD}}))}{c}$ and $p_{\mathsf{max}}= \frac{ f_{\sft,\mathsf{max}}d_\sfy (n_\sfg +\sin(\phi_{\mathsf{AoD}}))}{c}$.
	The optimal solution to problem  	\textbf{P3}  $f^{\star}_{\sft} $  is expressed in closed-form  as
\begin{align}\label{eqn: f_t given phi_t}
	f^{\star}_{\sft}(\phi_{\mathsf{AoD}})= \frac{p^{\star} c}{d_\sfy(\sin(\phi_{\mathsf{AoD}})+n_\sfg)},
\end{align}
where $p^{\star}\in [p_{\mathsf{min}}, p_{\mathsf{max}}]$ is the solution to the problem
\begin{align}\label{eqn: p star numerical soln}
	p^{\star}= \underset{p \in [p_{\mathsf{min}}, p_{\mathsf{max}}]}{\mathsf{argmax}}\left|\frac{\sin(\pi N p)}{\sin(\pi p)}\right|.
\end{align}
The  beamforming gain at $\phi_{\mathsf{AoD}}$ is $\cG^{\star}_{\mathsf{DMA}}(\phi_{\mathsf{AoD}}, f^{\star}_{\sft}(\phi_{\mathsf{AoD}})) \leq N^2$, where equality holds when 
$\exists~p^{\star}\in \bbN$. At this maximum value, the  optimal resonant frequencies of the DMA elements are $f_{\sfr, n}^{\star}(\phi_{\mathsf{AoD}})=f^{\star}_{\sft}(\phi_{\mathsf{AoD}})~\forall~n $. For the case where $\nexists  ~p^{\star}\in \bbN$, the optimal resonant frequencies are obtained by solving \eqref{eqn: p star numerical soln} numerically and substituting $f^{\star}_{\sft}(\phi_{\mathsf{AoD}})$ in \eqref{eqn: f star l closed form}.
\end{lemma}
\begin{proof}
		The inner maximization in problem \textbf{P3} is solved in Lemma~\ref{thm:  optimal soln P1a} and the closed-form resonant frequency solution is expressed in \eqref{eqn: f star l closed form}. We leverage these solutions to solve  problem \textbf{P3}. The optimal operating frequency is expressed as 
	\begin{align}
		f^{\star}_{\sft}(\phi)&= \underset{f_{\sft}  \in [f_{\sft,\mathsf{min}}, f_{\sft,\mathsf{max}}]}{\mathsf{argmax}}\cG^{\star}_{\mathsf{DMA}}(\phi, f_{\sft})= \underset{f_{\sft}  \in [f_{\sft,\mathsf{min}}, f_{\sft,\mathsf{max}}]}{\mathsf{argmax}} |	\cS(\phi, f_{\sft})|.
	\end{align}
	There are two cases for expressing the solution of \textbf{P3} depending on whether  $p^{\star}$ is a natural number or not.
	
	\textbf{Case 1:}	The global maximum of $|	\cS(\phi, f_{\sft})|$ is $N$ and is attained when $	p^{\star}= \frac{ f^{\star}_{\sft}d_\sfy (n_\sfg +\sin(\phi))}{c}$ is a natural number.
	Using \eqref{eqn: f star l closed form}, it can be easily verified that $	f_{\sfr, n}^{\star}(\phi)=f^{\star}_{\sft}(\phi)$ when $	p^{\star}$ is a natural number.
	%							To prove that $	f_{\sfr, n}^{\star}(\phi)=f^{\star}_{\sft}(\phi)$, there are three cases.
	%							\begin{enumerate}
		%								\item  If $N$ is an odd number, then $\frac{\sin(\pi N p^{\star})}{\sin(\pi p^{\star})}=N~\forall~p^{\star} \in \bbN$.  Hence, in \eqref{eqn: f star l closed form},  $\mathsf{sgn}(\cS(\phi, f_{\sft}))=1$ and $\tan(\pi 	p^{\star} (n-1- \frac{N-1}{2}) )=0~\forall~n$. 
		%								\item If $N$ is an even number, then $\frac{\sin(\pi N p^{\star})}{\sin(\pi p^{\star})}=N $ when $p^{\star}   $ is even.  Hence, in \eqref{eqn: f star l closed form},  $\mathsf{sgn}(\cS(\phi, f_{\sft}))=1$ and $\tan(\pi 	p^{\star} (n-1- \frac{N-1}{2}) )=0~\forall~n$. 
		%								\item If $N$ is an even number, then  $\frac{\sin(\pi N p^{\star})}{\sin(\pi p^{\star})}=-N $ when $p^{\star}   $ is odd.  Hence, in \eqref{eqn: f star l closed form},  $\mathsf{sgn}(\cS(\phi, f_{\sft}))=-1$ and $\tan(\pi/2+\pi 	p^{\star} (n-1- \frac{N-1}{2}) )=0~\forall~n$. 
		%							\end{enumerate}
	%							This completes the proof for \textbf{Case 1}.
	
	\textbf{Case 2:}  The global maximum $N$ cannot be attained if there is no integer $p^{\star}\in [p_{\mathsf{min}}, p_{\mathsf{max}}]$. Hence,  $p^{\star}$ is obtained numerically by solving \eqref{eqn: p star numerical soln}.
%The proof is in Appendix~\ref{proof: optimal target freq}.
\end{proof}

From Lemma~\ref{thm:  optimal target freq}, we see that the maximum possible DMA beamforming  gain is $N^2$ only when $p^{\star}$ is a natural number.  This  maximum gain  equals that of a TTD array but only at the optimal frequency. Lemma~\ref{thm:  optimal target freq} also demonstrates that all slots should have the same optimal resonant frequency for the DMA to attain the maximum gain.
The gain is less than $N^2$  when there does not exist any natural number in the range $[p_{\mathsf{min}}, p_{\mathsf{max}}]$. In this case, the optimal resonant frequencies are distinct for each DMA slot.

%The TTD and phased-array have a complete $360^{\circ}$ beamforming angular range because the maximum beamforming gain does not depend on the target angle and frequency.
The DMA angular range for beamforming is limited by the constraint on the  operating frequency in \eqref{eq: ft range}.
For the unconstrained DMA case, i.e., without constraint \eqref{eq: ft range}, the maximum beamforming gain is $N^2$ for all AoDs. 
  Imposing  constraint \eqref{eq: ft range} limits the angular range for which the maximum gain  $N^2$ is attainable. 
%  In Theorem~\ref{thm:  optimal target freq}, the DMA can attain this maximum beamforming gain only for a specific angular coverage for which the optimal target frequency lies inside the DMA operating frequency range.
The limits $p_{\mathsf{min}}$ and $p_{\mathsf{max}}$ play a critical role in defining the regime in which the maximum gain of  $N^2$ can be attained. These limits depend on the DMA parameters like the refractive index $n_\sfg$ and the inter-antenna spacing $d_\sfy$.  This dependence motivates us to design  $n_\sfg$ and $d_\sfy$ to improve the beamforming coverage.

We first define the maximum gain angular range of the DMA.
The maximum gain angular range is the set of angles over which the DMA beamforming gain is $N^2$ for some frequency $f^{\star}_{\sft}(\phi) \in  [f_{\sft,\mathsf{min}}, f_{\sft,\mathsf{max}}]$.
Let this set be denoted as $\mathbf{\Phi}$.
%Let the desired maximum gain angular range be $\phi\in [\phi_{\mathsf{l}}, \phi_{\mathsf{u}}]$. 
Note that the frequency at which the maximum gain is achieved is dependent on the angle as proved in Lemma~\ref{thm:  optimal target freq}.
%The beamforming criterion is that for each angle in this desired angular range, the DMA beamforming gain should attain the maximum value $N^2$ at the optimal target frequency.  
We define the beamforming criterion as 
\begin{align}\label{eqn: bf criterion}
\cG^{\star}_{\mathsf{DMA}}(\phi, f^{\star}_{\sft}(\phi))=N^2~\forall~\phi \in \mathbf{\Phi},  ~f^{\star}_{\sft}(\phi) \in  [f_{\sft,\mathsf{min}}, f_{\sft,\mathsf{max}}] .
\end{align}
%We focus on designing  $n_\sfg$ and $d_\sfy$ to satisfy a specific beamforming requirement. Let the desired beamforming angular range be $\phi_{\mathsf{AoD}} \in [-\phi_0, \phi_0]$. The objective is to have $\cG^{\star}_{\mathsf{DMA}}(\phi, f^{\star}_{\sft}(\phi))=N^2 \quad \forall \phi \in [-\phi_0, \phi_0] $.
 In other words, 
% we want to have the limits $p_{\mathsf{min}}$ and  $p_{\mathsf{max}}$ such that 
 for each  AoD  $ \phi \in \mathbf{\Phi}$, we want to obtain the optimal operating frequency solution corresponding to the case when $p^{\star}$ is a natural number from Lemma~\ref{thm:  optimal target freq}.
We provide the closed-form design solution for $n_\sfg^{\star}$ and $d_\sfy^{\star}$ to satisfy the beamforming criterion from \eqref{eqn: bf criterion} in the following lemma. 
%We design the inter-element spacing $d_\sfy^{\star}$ and the refractive index $n_\sfg^{\star}$   which leads to the target beamforming gain satisfy the beamforming criterion from \eqref{eqn: bf criterion}.

\begin{lemma}\label{thm:  ng and dy design}
	For the set $\mathbf{\Phi}= \{ \phi_{\mathsf{AoD}}: -\pi< \phi_{\sfl\sfw} \leq \phi_{\mathsf{AoD}} \leq \phi_{\sfu\sfp} < \pi\}$,
%	 Assuming $n_\sfg^{\mathsf{max}}$ is the maximum practically feasible waveguide refractive index,  
	 the closed-form design solution for $n_\sfg^{\star}$ that satisfies \eqref{eqn: bf criterion} is 
%	\begin{align}\label{eqn: optimal refractive index}
%		n_\sfg^{\star}= \mathsf{min}\left(\sin(\phi_0) \left(\frac{f_{\sft,\mathsf{max}}+f_{\sft,\mathsf{min}}}{f_{\sft,\mathsf{max}}-f_{\sft,\mathsf{min}}}\right), n_\sfg^{\mathsf{max}} \right),
%	\end{align}	
	\begin{align}\label{eqn: optimal refractive index}
n_\sfg^{\star}\!=\!\frac{\sin(\phi_{\sfu\sfp})\!-\!\sin(\phi_{\sfl\sfw})}{2}\!\! \left(\!\frac{f_{\sft,\mathsf{max}}\!+\!f_{\sft,\mathsf{min}}}{f_{\sft,\mathsf{max}}\!-\!f_{\sft,\mathsf{min}}}\!\right) \!-\! \frac{\sin(\phi_{\sfu\sfp})\!+\!\sin(\phi_{\sfl\sfw})}{2}\!,
	\end{align}
	and the closed-form design solution for $d_\sfy^{\star}$ is 
	\begin{align}\label{eqn: opt dy}
		d_\sfy^{\star}
%		&= \frac{c p^{\star}}{f_{\sft,\mathsf{min}} (	n_\sfg^{\star}+ \sin(\phi_0))}= \frac{c p^{\star}}{f_{\sft,\mathsf{max}} (	n_\sfg^{\star}- \sin(\phi_0))} \nonumber \\&
		= 
		\frac{c p^{\star} (f_{\sft,\mathsf{max}}- f_{\sft,\mathsf{min}})}{(\sin(\phi_{\sfu\sfp})-\sin(\phi_{\sfl\sfw})) f_{\sft,\mathsf{min}}f_{\sft,\mathsf{max}}}, \text{where }  p^{\star} \in \bbN.
	\end{align}
\end{lemma}
\begin{proof}
	From \textbf{Case 1} of  Lemma~\ref{thm:  optimal target freq}, we see that to satisfy the  requirement $\cG^{\star}_{\mathsf{DMA}}(\phi, f^{\star}_{\sft})=N^2  $,  the value of $p^{\star}$ has to be a natural number which lies inside $\left[ \frac{ f_{\sft,\mathsf{min}}d_\sfy (n_\sfg +\sin(\phi))}{c},  \frac{ f_{\sft,\mathsf{max}}d_\sfy (n_\sfg +\sin(\phi))}{c}\right]$. To satisfy $\cG^{\star}_{\mathsf{DMA}}(\phi, f^{\star}_{\sft})=N^2~\forall~\phi \in [\phi_{\sfl\sfw}, \phi_{\sfu\sfp}] $, the upper limit computed at $\phi \in [\phi_{\sfl\sfw}, \phi_{\sfu\sfp}] $ should be greater than or equal to a natural number $p^{\star}$ and the lower limit computed at $\phi \in [\phi_{\sfl\sfw}, \phi_{\sfu\sfp}]$ should be less than or equal to the same  natural number $p^{\star}$.
	Mathematically, $\forall \phi \in [\phi_{\sfl\sfw}, \phi_{\sfu\sfp}]$, we get
	%							\begin{align}
		%								\frac{ f_{\sft,\mathsf{min}}}{c}d_\sfy (n_\sfg +\sin(\phi))\bigg|_{\forall \phi \in [\phi_{\sfl}, \phi_{\sfu}]} \leq p^{\star} \leq \frac{ f_{\sft,\mathsf{max}}d_\sfy (n_\sfg +\sin(\phi))}{c}\bigg|_{\forall \phi \in [\phi_{\sfl}, \phi_{\sfu}]} .
		%							\end{align}
	\begin{align}
		\frac{ f_{\sft,\mathsf{min}}d_\sfy}{c} (n_\sfg +\sin(\phi)) \leq p^{\star} \leq \frac{ f_{\sft,\mathsf{max}}d_\sfy}{c} (n_\sfg +\sin(\phi)).
	\end{align}
	%							$	\frac{ f_{\sft,\mathsf{min}}d_\sfy (n_\sfg +\sin(\phi))}{c}\bigg|_{\forall \phi \in [\phi_{\sfl}, \phi_{\sfu}]} \leq p^{\star} \leq \frac{ f_{\sft,\mathsf{max}}d_\sfy (n_\sfg +\sin(\phi))}{c}\bigg|_{\forall \phi \in [\phi_{\sfl}, \phi_{\sfu}]} .$
	%							\begin{align}
		%								&\frac{ f_{\sft,\mathsf{min}}d_\sfy (n_\sfg +\sin(\phi))}{c}\bigg|_{\forall \phi \in [\phi_{1}, \phi_{2}]} \leq p^{\star} \\\nonumber &\leq \frac{ f_{\sft,\mathsf{max}}d_\sfy (n_\sfg +\sin(\phi))}{c}\bigg|_{\forall \phi \in [\phi_{1}, \phi_{2}]} .
		%							\end{align}
	Both the upper and lower limit on $p^{\star}$ are an increasing function of $\phi$ in the desired angular range.  Hence, we have
	\begin{align}\label{eqn: upper and lower on p star}
		\frac{ f_{\sft,\mathsf{min}}d_\sfy (n_\sfg +\sin(\phi_{\sfu\sfp}))}{c} \leq  p^{\star} \leq  \frac{ f_{\sft,\mathsf{max}}d_\sfy (n_\sfg +\sin(\phi_{\sfl\sfw}))}{c}.
	\end{align}
	%	Let there be a refractive index $n_\sfg=	n_\sfg^{\star} $ such that equality holds.
	By equating the upper and lower limit of $p^{\star}$ in \eqref{eqn: upper and lower on p star}, we get the desired results in \eqref{eqn: optimal refractive index} and \eqref{eqn: opt dy} where equality holds for $n_\sfg=	n_\sfg^{\star} $ and $d_\sfy=	d_\sfy^{\star}$.
%	The proof is in Appendix~\ref{proof: ng and dy design}.
\end{proof}

The beamforming criterion in \eqref{eqn: bf criterion} can be satisfied by an optimal design choice of the refractive index and element spacing for the DMA.
For the special case of a symmetric angular range with $\phi_{\mathsf{max}}$ as the maximum desired beamforming angle from broadside ($\phi_{\mathsf{AoD}}=0$), let $\phi_{\sfl\sfw}=-\phi_{\mathsf{max}}$ and $\phi_{\sfu\sfp}=\phi_{\mathsf{max}}$.  This leads to $	n_\sfg^{\star}= \sin(\phi_{\mathsf{max}}) \left(\frac{f_{\sft,\mathsf{max}}+f_{\sft,\mathsf{min}}}{f_{\sft,\mathsf{max}}-f_{\sft,\mathsf{min}}}\right) $ and $	d_\sfy^{\star}=\frac{c p^{\star} (f_{\sft,\mathsf{max}}- f_{\sft,\mathsf{min}})}{2  \sin(\phi_{\mathsf{max}}) f_{\sft,\mathsf{min}}f_{\sft,\mathsf{max}}}$.
From \textbf{Case 1} of Lemma~\ref{thm:  optimal target freq}, we see that the maximum value of beamforming gain $N^2$ is obtained when the value of  $p^{\star}$ is chosen to be a natural number. 
As $d_\sfy^{\star}$ depends on  $p^{\star}$,  a reasonable choice of $p^{\star}$ should be made such that the spacing is not too large to avoid grating lobes and not too small to avoid mutual coupling between different elements. 

The  angular coverage is limited by the available choices of the waveguide substrate. Let $n_\sfg^{\mathsf{max}}$ be the maximum practically feasible waveguide refractive index.
This imposes the constraint
\begin{align}\label{eqn:  condition for max angular range}
n_\sfg^{\star}=	\sin(\phi_{\mathsf{max}}) \left(\frac{f_{\sft,\mathsf{max}}+f_{\sft,\mathsf{min}}}{f_{\sft,\mathsf{max}}-f_{\sft,\mathsf{min}}}\right)\leq n_\sfg^{\mathsf{max}}.
\end{align}
%Let the maximum feasible beamforming angle from the broadside to attain the maximum gain $N^2$ be $\phi_{\mathsf{max}}$.  
Let the  frequency range over which the DMA slots can be tuned be $\sfT_{\sfr}=f_{\sft,\mathsf{max}}-f_{\sft,\mathsf{min}}$.  For a center frequency of the tuning range $f_{\sfc}$, let $f_{\sft,\mathsf{max}}= f_{\sfc}+\frac{\sfT_{\sfr}}{2}$ and $f_{\sft,\mathsf{min}}= f_{\sfc}-\frac{\sfT_{\sfr}}{2}$.
From \eqref{eqn:  condition for max angular range}, we compute 
$\phi_{\mathsf{max}}$ as 
\begin{align}\label{eqn: phi max}
\phi_{\mathsf{max}} = \sin^{-1}\left(   \frac{n_\sfg^{\mathsf{max}} \sfT_{\sfr}}{2 f_{\sfc}}\right).
\end{align}
From \eqref{eqn: phi max}, we observe that to maximize the DMA beamforming angular coverage, the DMA should use the highest feasible waveguide refractive index and the maximum possible operating range $\sfT_{\sfr}$ for the given $f_{\sfc}$.
In Fig.~\ref{fig: max beamforming angle vs fractional operating band}, we plot $\phi_{\mathsf{max}}$ as a function of the fractional operating band $\frac{\sfT_{\sfr} }{f_{\sfc}}$ for  the cases $n_\sfg=2.5$ and $n_\sfg=4$.
A typical base station serves a sector of $60^{\circ}$, i.e., $\phi_{\mathsf{max}}=30^{\circ}$.  For this value of $\phi_{\mathsf{max}}$, the required fractional operating band is $\frac{\sfT_{\sfr} }{f_{\sfc}}= 0.25$ for $n_\sfg^{\mathsf{max}}=4$ as marked in  Fig.~\ref{fig: max beamforming angle vs fractional operating band}. This design value of the refractive index is feasible. For example, Germanium has a refractive index of 4.

	\begin{figure}
	\centering
		\includegraphics[width=0.4\textwidth]{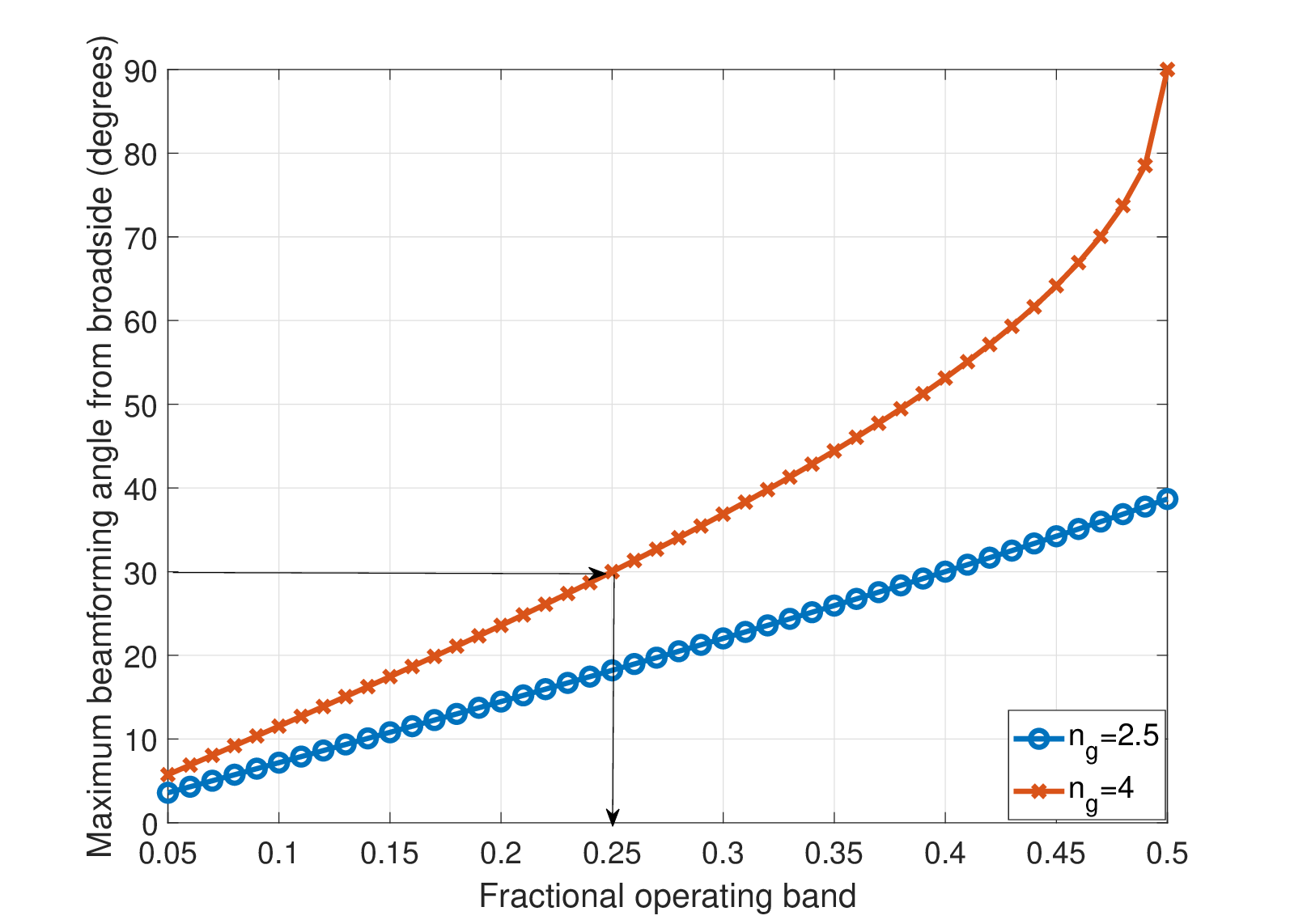}    
		\caption
		{The DMA maximum gain angular range from broadside is plotted as a function of  the fractional operating band for different refractive indices. The maximum beamforming angle can be improved using wider operating band and higher refractive
			index.
		}
		\label{fig: max beamforming angle vs fractional operating band}
\end{figure}

				\subsection{DMA beamforming gain frequency response analysis}	\label{subsec: DMA beamforming gain frequency response and half-power bandwidth analysis}

			So far, we have demonstrated how the DMA can achieve a gain of $N^2$ at AoD $\phi_{\mathsf{AoD}}$ only at a particular frequency $f^{\star}_{\sft}(\phi_{\mathsf{AoD}})$. The data transmission, however, is for a specific bandwidth $\mathsf{B}$.
		A phased-array experiences degradation in the beamforming gain for frequencies away from the center frequency due to  beam-squint~\cite{10002944}. While DMA beamforming also exhibits beam-squint, the DMA beamforming gain experiences further wideband losses because of the  per-element frequency-selective Lorentzian response shown in \eqref{eqn: Lorentzian form}.

		 We analyze the frequency response of the DMA at an AoD inside the  desired angular range for a bandwidth around the optimal operating frequency. 
			For a DMA designed using the refractive index and element spacing from Lemma~\ref{thm:  ng and dy design},  the maximum beamforming gain $N^2$ is attained at  $	f^{\star}_{\sft}(\phi)$  for any AoD $\phi $ in the desired angular range $ [\phi_{\sfl\sfw}, \phi_{\sfu\sfp}]$. To attain the maximum beamforming gain, the resonant frequencies are tuned based on Lemma~\ref{thm:  optimal target freq}, i.e.,  $	f_{\sfr, n}^{\star}(\phi)=f^{\star}_{\sft}(\phi), \quad \forall~n~\in~\{1, \dots, N\}$. Hence, we have $	[\bff_{\mathsf{DMA}}(f)]_{n}= \frac{\Gamma f}{2\pi (f^{\star}_{\sft}(\phi))^2- 2\pi f^2 +\sfj \Gamma f }, \quad \forall n\in\{1, \dots, N\}.$
		%				\begin{align}
			%					[\bff_{\mathsf{DMA}}(f)]_{n}= \frac{\Gamma f}{2\pi (f^{\star}_{\sft}(\phi))^2- 2\pi f^2 +\sfj \Gamma f }, \quad \forall n=\{1, \dots, N\}.
			%				\end{align}
		Substituting $[\bff_{\mathsf{DMA}}(f)]_{n}$ in \eqref{eqn: G dma}, we have 
		\begin{align}\label{eqn:  G_dma_phit_f_eqn_in_proof}
			\cG_{\mathsf{DMA}}(\phi, f)= \bigg|  \frac{\Gamma f}{2\pi (f^{\star}_{\sft}(\phi))^2- 2\pi f^2 +\sfj \Gamma f }\bigg|^2 \bigg| \bm{1}^T \bh(\phi, f)    \bigg|^2.
		\end{align}
		To analyze the overall DMA beamforming response, we write $\cG_{\mathsf{DMA}}(\phi, f)$ as product of the per-element frequency response  $\cG_{\mathsf{element}}(\phi, f)= \left|  \frac{\Gamma f}{2\pi (f^{\star}_{\sft}(\phi))^2- 2\pi f^2 +\sfj \Gamma f }\right|^2 $ and the frequency response of the array of DMA slots $\cG_{\mathsf{array}}(\phi, f)=\left| \bm{1}^T \bh(\phi, f)    \right|^2 $.
				 We can verify that the maximum value of $N^2$ is attained at $f=f^{\star}_{\sft}(\phi)$.
%				  by substituting \eqref{eqn: f_t given phi_t} for $f$ in \eqref{eqn:  G_dma_phit_f_eqn_in_proof}. 
				   For $f\neq f^{\star}_{\sft}(\phi)$, the beamforming gain $\cG_{\mathsf{DMA}}(\phi, f) < N^2$ as $ \cG_{\mathsf{element}}(\phi, f) <1$ and $\cG_{\mathsf{array}}(\phi, f) <~ N^2$.
				   
			To analyze the  DMA bandwidth for which the beamforming gain does not fall below a  certain threshold of $\nu$ times the maximum,  we define the $\nu$ cutoff frequencies.
			Let the upper $\nu$ cutoff frequency be the highest frequency for which  $\cG_{\mathsf{DMA}}(\phi, f) \geq \nu N^2$ and be denoted as $f^{\sfu}_{\nu} $. Let the lower  $\nu$ cutoff frequency be the lowest frequency for which  $\cG_{\mathsf{DMA}}(\phi, f) \geq \nu N^2$ and be denoted as $f^{\sfl}_{\nu} $. 
%			We define the frequency interval $[ f^{\sfl}_{\nu} , f^{\sfu}_{\nu} ]$ where 
%			   $f^{\sfu}_{\nu}>f^{\star}_{\sft}(\phi) $ is the upper $\nu$ cutoff frequency and  $f^{\sfl}_{\nu}< f^{\star}_{\sft}(\phi)$ is the lower $\nu$ cutoff frequency. 
			%  For each frequency $f$ within this interval, we have $\cG_{\mathsf{DMA}}(\phi, f) \geq \nu N^2$, where $\nu<1$.
%			   At both $\nu$ cutoff frequencies, the beamforming gain is $\nu$ times the  peak value, i.e.,
%			$\cG_{\mathsf{DMA}}(\phi, f^{\sfu}_{\nu})=\cG_{\mathsf{DMA}}(\phi, f^{\sfl}_{\nu})=\nu N^2, $ where $ \nu<1.$
%			\begin{align}
%			\cG_{\mathsf{DMA}}(\phi, f^{\sfu}_{\nu})=\cG_{\mathsf{DMA}}(\phi, f^{\sfl}_{\nu})=\nu N^2,  \text{ where }  \nu<1.
% 			\end{align}
			  Let us define the $\nu$ bandwidth as $\mathsf{B}_{\nu}~=~f^{\sfu}_{\nu}~-~f^{\sfl}_{\nu}.$
			  %or the half-power bandwidth as 
%			\begin{align}
%				\mathsf{B}_{\nu}= f^{\sfu}_{\nu} -f^{\sfl}_{\nu}.
%			\end{align}
We first analyze the cutoff frequencies  and bandwidth for $\cG_{\mathsf{element}}(\phi, f)$ in the following lemma.
	Let $f^{\sfu}_{\sfee,\nu} $, $f^{\sfl}_{\sfee,\nu} $, and $\mathsf{B}_{\sfee,\nu}$ be the upper cutoff frequency, lower cutoff frequency, and the bandwidth for the term $\cG_{\mathsf{element}}(\phi, f)$.
		%	The following lemma gives an approximation for $f^{\sfl}_{\nu} $, $f^{\sfu}_{\nu} $, and $\mathsf{B}_{\nu}$.
 		
 				\begin{lemma}\label{thm: DMA 3 dB bandwidth}
 					Let $\rho(\nu)=\frac{\nu}{1-\nu}$.
 					For $\cG_{\mathsf{element}}(\phi, f)$, the closed-form exact expression   for the   $\nu$ cutoff frequencies is 
 					\begin{subequations}
 						\begin{equation}\label{eqn: closed form f 3db lower}
 								f^{\sfl}_{\sfee,\nu} = \sqrt{(f^{\star}_{\sft}(\phi))^2+\frac{\Gamma^2-\sqrt{\Gamma^4+16 \Gamma^2 \pi^2 (f^{\star}_{\sft}(\phi))^2\rho(\nu)}}{8\pi^2\rho(\nu)}},
 							\end{equation}
 				\begin{equation}\label{eqn: closed form f 3db upper}
 				f^{\sfu}_{\sfee,\nu} = \sqrt{(f^{\star}_{\sft}(\phi))^2+\frac{\Gamma^2+\sqrt{\Gamma^4+16 \Gamma^2 \pi^2 (f^{\star}_{\sft}(\phi))^2\rho(\nu)}}{8\pi^2\rho(\nu)}}.
 				\end{equation}
 					\end{subequations}
 				The $\nu$ bandwidth for $\cG_{\mathsf{element}}(\phi, f)$ is expressed as
 					\begin{align}\label{eqn: B 3dB closed-form}
 							\mathsf{B}_{\sfee,\nu}=\frac{\Gamma}{2\pi \sqrt{\rho(\nu)}}+ \cO(\Gamma^3).
 					\end{align} 
 				\end{lemma}
				\begin{proof}
			Mathematically, the  closed-form expression for the $\nu$ cutoff frequencies are obtained as a solution to $\frac{\Gamma^2 f^2}{(2\pi (f^{\star}_{\sft}(\phi))^2- 2\pi f^2 )^2+ \Gamma^2 f^2}=\nu.$
			%to the following equation in $f$.
			%							\begin{align}\label{eqn:  condition for 3 dB cutoff}
				%								\frac{\Gamma^2 f^2}{(2\pi (f^{\star}_{\sft}(\phi))^2- 2\pi f^2 )^2+ \Gamma^2 f^2}=\nu.
				%							\end{align}
			We express this equation as a quadratic in $f^2$.
			%					 as follows.
			%					\begin{align}\label{eqn: quadratic in f2^2}
				%						4\pi^2 (f^2)^2 + \left( -\frac{\Gamma^2}{\rho(\epsilon)}-8\pi^2 (f^{\star}_{\sft}(\phi))^2 \right)f^2 +4\pi^2  (f^{\star}_{\sft}(\phi))^4=0
				%					\end{align}
			Solving the quadratic,
			%in \eqref{eqn: quadratic in f2^2}, 
			we obtain the desired cutoff frequencies in \eqref{eqn: closed form f 3db lower} and \eqref{eqn: closed form f 3db upper}. 
			%							These expressions are valid as long as   the assumption of ignoring the first term in \eqref{eqn: G DMA phit f closed form} holds.  For smaller values of $N$, this assumption is valid because there is negligible beam squint.
			The $\nu$ bandwidth closed-form expression in \eqref{eqn: B 3dB closed-form} is obtained by applying binomial expansion to \eqref{eqn: closed form f 3db lower} and \eqref{eqn: closed form f 3db upper}.
				\end{proof}
				
				The frequency response of $\cG_{\mathsf{element}}(\phi, f)$ has sharp degradation with frequency compared to  $ \cG_{\mathsf{array}}(\phi, f) $. 	Hence, we use the result from Lemma~\ref{thm: DMA 3 dB bandwidth} directly to obtain $f^{\sfl}_{\nu} \approx	f^{\sfl}_{\sfee,\nu} $, $f^{\sfu}_{\nu} \approx	f^{\sfu}_{\sfee,\nu} $, and $ 	\mathsf{B}_{\nu} \approx \mathsf{B}_{\sfee,\nu}$. 
					The key design insight is that the $\nu$  bandwidth $	\mathsf{B}_{\mathsf{\nu}}$ is  directly proportional to the DMA damping factor $\Gamma$. For example,  for 3 dB bandwidth, i.e, $\nu=0.5$, we have $\rho(0.5)=1$ and $	\mathsf{B}_{0.5}\approx \frac{\Gamma}{2\pi}$.
%				Higher damping factor leads to increase in DMA bandwidth. The quality factor $	\sfQ $ also decreases for higher damping factor which affects the radiation efficiency of the DMA.
%			Assuming the transmit power spectral density  $	P_{\mathsf{T}}(f)$ is fixed, 	the DMA communication performance can be improved by increasing the damping factor because it leads to higher bandwidth in the beamforming gain response.
%			A direction of future work is to  analyze the tradeoff between radiation efficiency and  beamforming gain through different damping factors. 
		
%				 The research on how to design  DMAs with higher damping factor is a direction for future work. 
%				
					\begin{figure}
					\centering
					\includegraphics[width=0.4\textwidth]{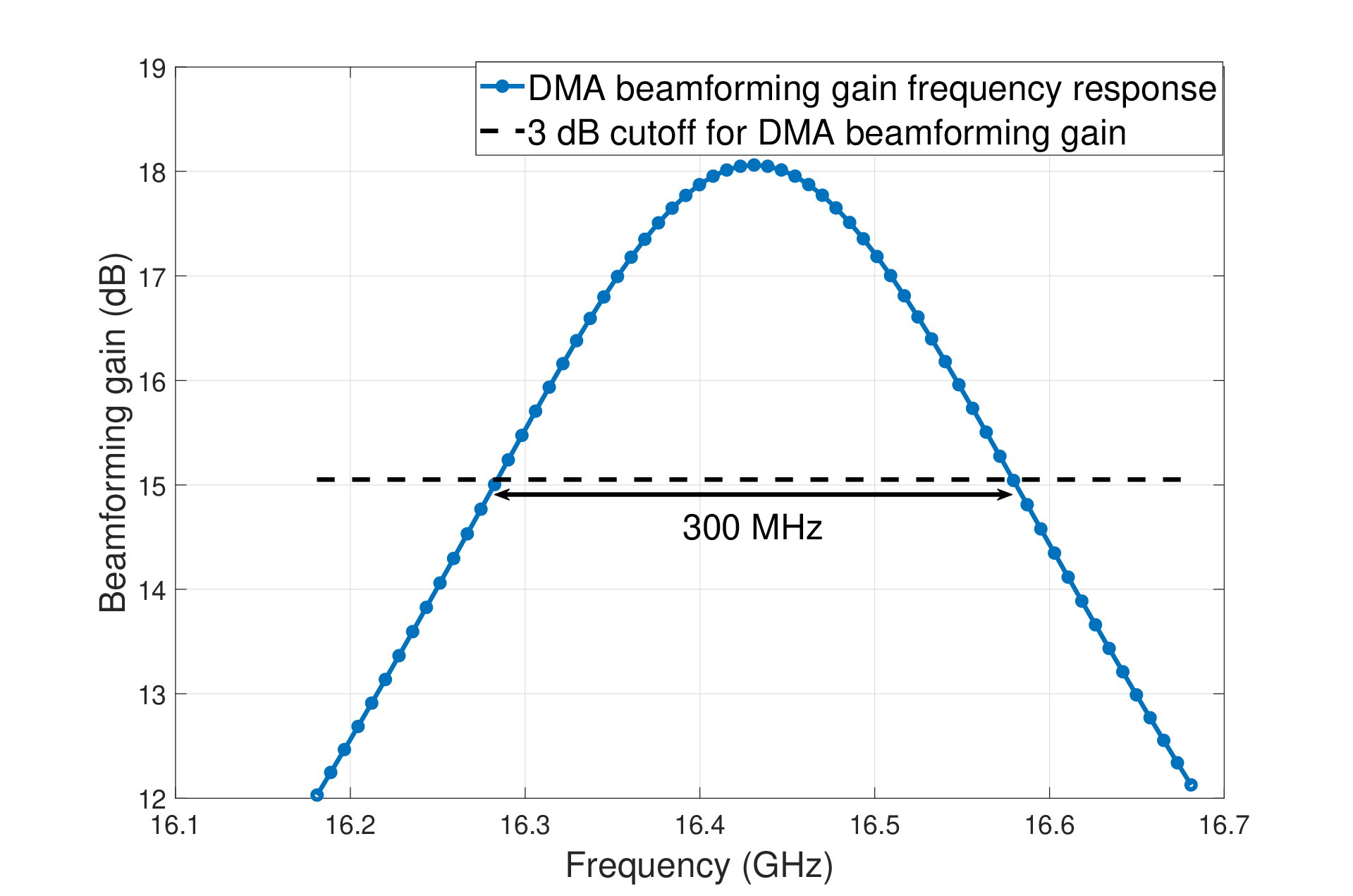}    
					\caption
					{The beamforming gain of the DMA configured to maximize gain at $\phi=-18^{\circ}$ and $	f^{\star}_{\sft}(-18^{\circ})= 16.43$ GHz is shown as a function of frequency.
					The gain	 decreases for frequencies away from the optimal operating frequency. The 3 dB bandwidth of the beamforming gain frequency response is proportional to the damping factor.
					}
					\label{fig: Beamforming_gain_freq_response_3dB}
				\end{figure}

\subsection{Numerical results for DMA beamforming gain}\label{subsec: numerical results bf gain}

					We verify the closed-form expressions from Lemma~\ref{thm: DMA 3 dB bandwidth} through numerical simulations. Let the number of elements on the DMA be $N=8$.  Let the center frequency of the DMA operating band be $f_{\sfc}=15$ GHz and the center wavelength $\lambda_{\sfc}=\frac{c}{f_{\sfc}}$.
					Let the DMA operating range be  $f_{\sft,\mathsf{min}}=12$ GHz and $f_{\sft,\mathsf{max}}=18$ GHz, and let the  desired angular range be $[-30^{\circ}, 30^{\circ}]$. From Lemma~\ref{thm:  ng and dy design}, we have the optimal spacing as
					 $d_\sfy^{\star}=0.42 \lambda_{\sfc}$ and the optimal refractive index $n_\sfg^{\star}=2.5$.  The  angular direction is set to  $\phi=-18^{\circ}$, which lies inside the desired angular range $[-30^{\circ}, 30^{\circ}]$. Let $\Gamma=\frac{2\pi f_{\sfc}}{50}$. To maximize the beamforming gain at $\phi$, we set the operating frequency and the DMA resonant frequencies using the expression in   Lemma~\ref{thm:  optimal target freq} with choice of $p^{\star}=1$. Using these parameters, we obtain the operating frequency $	f^{\star}_{\sft}(-18^{\circ})= 16.43$ GHz from \eqref{eqn: f_t given phi_t}. 
					 In Fig.~\ref{fig: Beamforming_gain_freq_response_3dB}, we plot the frequency response of the DMA beamforming gain at $\phi=-18^{\circ}$. As expected, the DMA beamforming gain frequency response curves peaks at the operating frequency 16.43 GHz.  
%					 We mark the 3 dB cutoff frequency points. 
					 From  Fig.~\ref{fig: Beamforming_gain_freq_response_3dB},  the 3 dB bandwidth equals 300 MHz which validates Lemma~\ref{thm: DMA 3 dB bandwidth}.

					We compare our proposed optimal operating frequency approach with methods used in the prior work.  
				\subsubsection{Fixed operating frequency approach}
				Prior work on DMA beamforming is based on optimizing the beamforming gain for a specific AoD and fixed operating frequency\cite{carlson2023hierarchical}.  
				Let the operating frequency be fixed to the center frequency $f_{\sfc}$.
				The DMA resonant frequencies are configured using the closed-form expression in \eqref{eqn: f star l closed form} where $f_{\sft}=f_{\sfc}$.
				In Fig.~\ref{fig: bf_gain_15G_binary_vs_continuous_tuning_center_vs_target}, we plot the optimal DMA beamforming gain at $f_{\sfc}$ as a function of the  angle $\phi$, mathematically denoted as $\cG^{\star}_{\mathsf{DMA}}(\phi, f_{\sfc})$ and computed using  \eqref{eqn: Gstar DMA phi t ft}. 
%				We show that for only one angle, the DMA beamforming gain obtained from the fixed operating frequency is optimal and matches our proposed approach. For all other  angles, the fixed frequency beamforming approach  leads to a beamforming gain lower than the maximum beamforming gain of $N^2$ as achieved by our proposed approach.
%				We show that for all target angles in the desired angular coverage except one target angle, the DMA beamforming gain obtained from the fixed operating frequency approach is sub-optimal, i.e., lower than the maximum beamforming gain of $N^2$ as achieved by our proposed approach.
				 We observe that $\cG^{\star}_{\mathsf{DMA}}(\phi, f_{\sfc}) \leq \cG^{\star}_{\mathsf{DMA}}(\phi, f^{\star}_{\sft} (\phi)), \quad \forall \phi $. The equality holds only for a single  angle denoted as $\phi_{\sfc}$ such that $f^{\star}_{\sft} (\phi_{\sfc})= f_{\sfc}$. Mathematically, it is expressed as 
				\begin{align}\label{eqn phi c}
					\phi_{\sfc}=\sin^{-1}\left(\frac{c}{f_{\sfc} d_\sfy}-n_\sfg\right).
				\end{align}
				In Fig.~\ref{fig: bf_gain_15G_binary_vs_continuous_tuning_center_vs_target}, we use the simulation parameters same as   in Fig.~\ref{fig: Beamforming_gain_freq_response_3dB}.
				In this case, \eqref{eqn phi c} evaluates to $\phi_{\sfc}=-5.74^{\circ}$.
				In Fig.~\ref{fig: bf_gain_15G_binary_vs_continuous_tuning_center_vs_target},	 we see that 
				$\cG^{\star}_{\mathsf{DMA}}(\phi_{\sfc}, f_{\sfc}) =N^2 $.
				For all angles $\phi \neq \phi_{\sfc}$, we see that $\cG^{\star}_{\mathsf{DMA}}(\phi, f_{\sfc}) < N^2$. 
				In Fig.~\ref{fig: optimal_target_freq_target_angle}, we plot the optimal frequency $f^{\star}_{\sft} (\phi)$ from \eqref{eqn: f_t given phi_t} as function of $\phi$.
				By optimally choosing the operating frequency based on the  angle as shown in Fig.~\ref{fig: optimal_target_freq_target_angle}, we significantly outperform the fixed operating frequency approach as indicated in Fig.~\ref{fig: bf_gain_15G_binary_vs_continuous_tuning_center_vs_target} which shows a gain of $N^2$ in the desired angular range.
%				 We also show that $\phi_{\sfc}=-5.74^{\circ}$ corresponds to the operating frequency of 15 GHz.
				
			The maximum gain angular range depends on the DMA upper and lower tuning frequency limits.
				 As $f_{\sft,\mathsf{min}}=12$ GHz and $f_{\sft,\mathsf{max}}=18$ GHz, the maximum beamforming angle for $n_\sfg=2.5$ is $\pm 30^{\circ}$. Beyond $\pm 30^{\circ}$, we observe a clipping  in the optimal operating frequency curve in Fig.~\ref{fig: optimal_target_freq_target_angle} to satisfy the operating frequency constraint \eqref{eq: ft range}.  This also reflects in the beamforming gain plot in Fig.~\ref{fig: bf_gain_15G_binary_vs_continuous_tuning_center_vs_target}, where we observe a  decrease in the beamforming gain outside $\pm 30^{\circ}$. 
				Without the DMA constraint~\eqref{eq: ft range}, the beamforming gain will be $N^2$ for the whole angular range.
				This is practically infeasible as the DMA would need to operated over a very wide frequency range of 10 GHz to 24 GHz.
			This is demonstrated	in Fig.~\ref{fig: optimal_target_freq_target_angle},
				where the monotonic curve extends beyond the feasible frequency range.

				\begin{figure}
					\centering
					\includegraphics[width=0.5\textwidth]{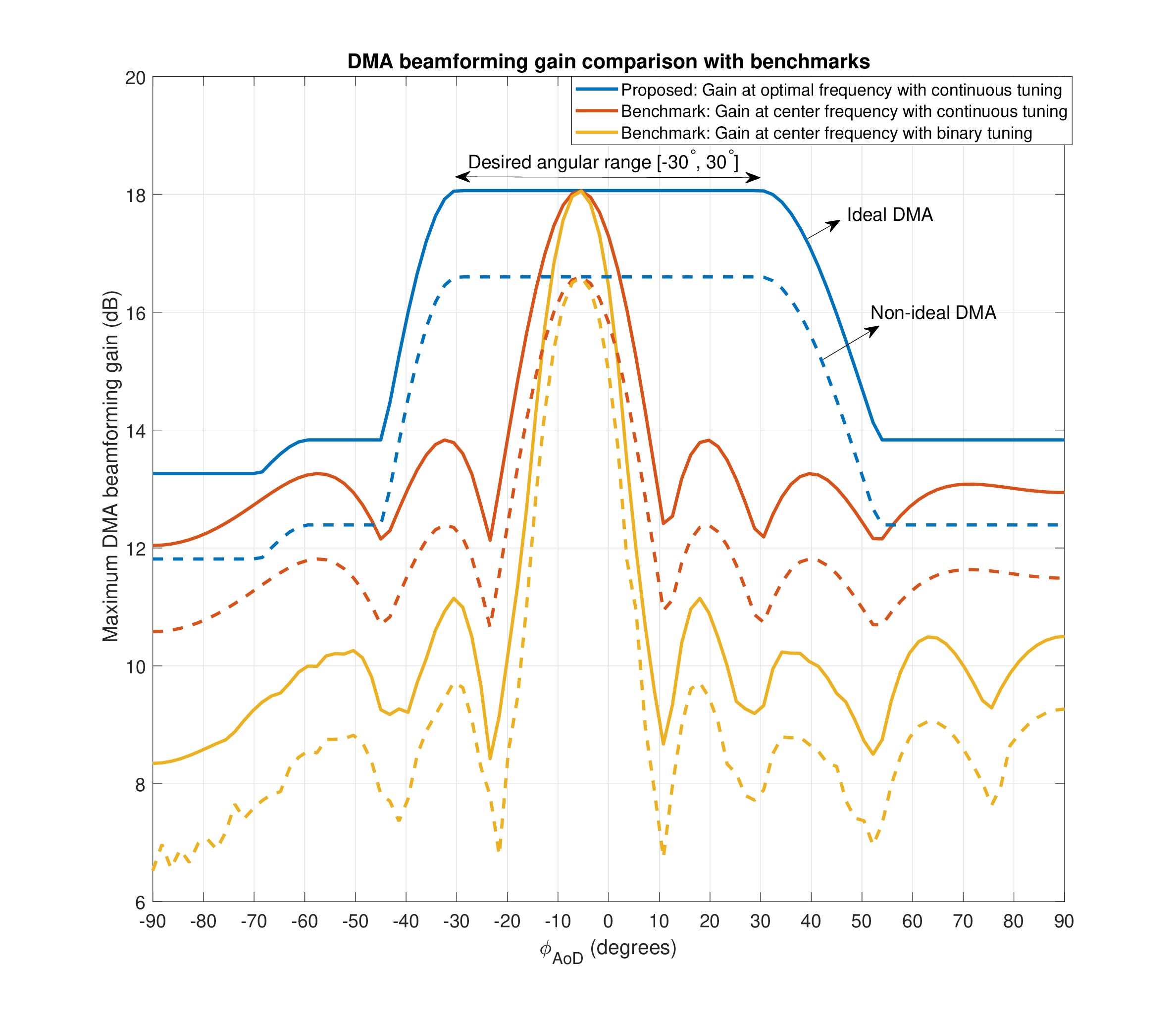}    
					\caption
					{The DMA beamforming gain as a function of the AoD is shown for both ideal and non-ideal DMA for different benchmarks.
					The	proposed approach outperforms the benchmarks with a fixed operating frequency for all  angles.  Beyond $\pm 30^{\circ}$,  the beamforming gain decreases because of the operating frequency constraint \eqref{eq: ft range}. 
					}
					\label{fig: bf_gain_15G_binary_vs_continuous_tuning_center_vs_target}
				\end{figure}
				
				\begin{figure}
					\centering
					\includegraphics[width=0.45\textwidth]{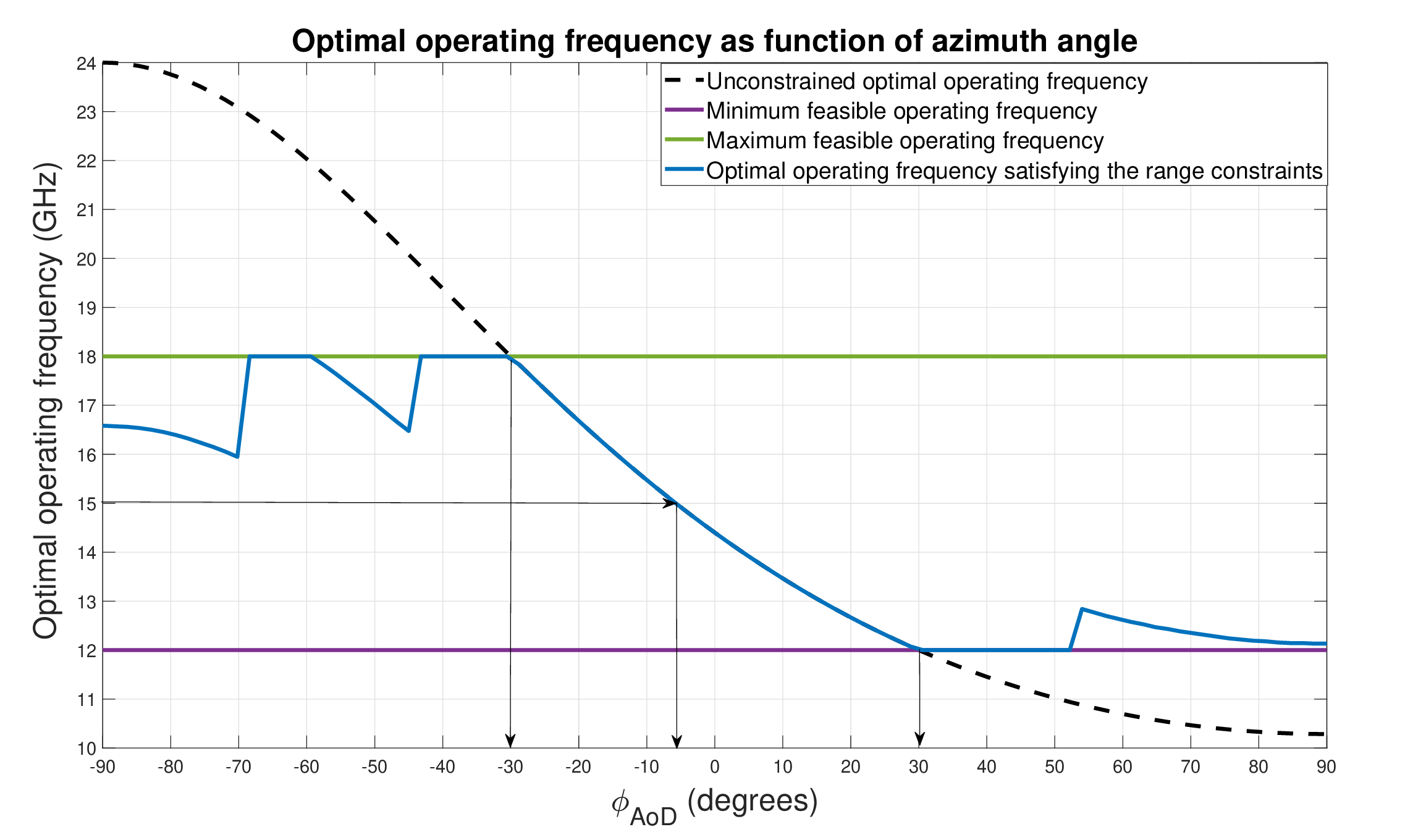}    
					\caption
					{The optimal operating frequency from our proposed approach is shown as a function of the AoD.
				A monotonic behavior for the desired angular range $\pm 30^{\circ}$ is observed which corresponds to the maximum gain in Fig.~\ref{fig: bf_gain_15G_binary_vs_continuous_tuning_center_vs_target}.
					}
					\label{fig: optimal_target_freq_target_angle}
				\end{figure}

				\subsubsection{Comparison between binary and continuous weights}

				In the analysis so far, we have assumed that DMA resonant frequencies are tunable to  a continuous frequency value.
%				This requires a high resolution data converter to obtain continuous tuning for the varactor diode.  
				{Another option for configuring the DMA elements is  a binary tuning mechanism leveraging PIN diodes instead of varactor diodes \cite{LinEtAlHighEfficiencyReconfigurableElementDynamic2020}. PIN diodes can be used as RF switches with on and off states, as opposed to the continuous varactor diode tuning. While the PIN diode devices themselves consume more power than passive varactor diodes, varactor-based DMAs require high-resolution DACs for continuous tuning and can potentially consume more power than a PIN diode-based DMA system \cite{WangEtAlReconfigurableIntelligentSurfacePower2024,TrichopoulosEtAlDesignEvaluationReconfigurableIntelligent2022}. Moreover, prior signal processing work on DMAs has highlighted the ability of binary tuning strategies to closely approach the performance of a continuous tuning strategy in narrowband settings \cite{ShlezingerEtAlDynamicMetasurfaceAntennasUplink2019,WangEtAlDynamicMetasurfaceAntennasMIMOOFDM2021,9847609}.}			
%					We compare with the  PIN diode-based DMA benchmark that leads to  binary weights, i.e., zeros and ones, as opposed to  continuous complex beamforming weights. 
				The binary weight beamforming problem is formulated as 
				\begin{subequations}\label{problem5}
					\begin{alignat}{3}
						\textbf{P4}:& \underset{\bff_{\mathsf{DMA}}(f_{\sfc})}{\mbox{ max }} \cG_{\mathsf{DMA}}(\phi, f_{\sfc})
						, \\
						&\text{ s.t. } [\bff_{\mathsf{DMA}}(f_{\sfc})]_{n}\in \{0,1\}.
						\label{problem5_a}
					\end{alignat}
				\end{subequations}
			%From the constraint, we see that	\textbf{P4} is a discrete optimization problem.
			 We solve \textbf{P4} by maximizing over the set of feasible beamforming vectors by brute force search.  The optimal beamforming gain at the  center frequency obtained after solving  \textbf{P4} is shown in  Fig~\ref{fig: bf_gain_15G_binary_vs_continuous_tuning_center_vs_target}. We see that when $\phi = \phi_{\sfc}$, the DMA beamforming gain for binary and continuous weights are equal to the theoretical maximum $N^2$. This observation is consistent with the result from Lemma~\ref{thm:  optimal target freq} that at the optimal operating frequency, the resonant frequencies on all elements are same as the operating frequency.
				For all  angles $\phi \neq \phi_{\sfc}$,  we see that the beamforming gain from binary weights is decreased compared to the continuous DMA beamforming weights.  For angles outside $[-14^{\circ}, 2^{\circ}]$, the decrease is higher than 2 dB.  This is because some of the antenna elements have beamforming weight as zero, i.e., they do not radiate, which reduces the effective number of antennas and hence the beamforming gain.

				\subsubsection{Non-ideal DMA}
				In the simulations so far, we have ignored the attenuation within the waveguide. For a practical DMA, the waveguide attenuation can degrade the beamforming gain~\cite{carlson_qif1}. To analyze the impact of waveguide attenuation, we denote the attenuation factor as $\alpha$ and assume an exponentially  decaying attenuation vector as $\bg(f)= [1, e^{-\alpha d_\sfy}, \dots, e^{-\alpha (N-1)d_\sfy  }]^T$.  The effective propagation channel is modified  as
				$	\bh(\phi,f)=~\bh_{\mathsf{dma}}(f) \odot 	\ba(\phi, f) \odot \bg(f).$
				For simulations, we use $\alpha =6$\cite{smith_analysis_2017} and use the resonant frequency solution corresponding to the ideal DMA case.
				In Fig.~\ref{fig: bf_gain_15G_binary_vs_continuous_tuning_center_vs_target}, we compare the beamforming gain of the ideal DMA and the non-ideal DMA. We see that the non-ideal DMA with attenuation exhibits a slight degradation in the beamforming gain compared to the ideal DMA without attenuation. We also see that the proposed optimal frequency approach outperforms the benchmark fixed frequency approach for the non-ideal DMA too. 
				Although inclusion of attenuation in the model leads to beamforming gain degradation, the overall comparison trend remains the same. Hence, we continue with the use of the ideal DMA model for subsequent analysis.
				%This validates the applicability of our approach for a practical DMA too.

				\section{Single-shot beam training with DMAs}\label{sec: Single-shot beam training with DMAs}

				The analysis in Section~\ref{sec: Frequency-selective beamforming gain optimization and analysis} assumes that the transmitter knows the AoD $\phi$.
%			For the analysis in Section~\ref{sec: Frequency-selective beamforming gain optimization and analysis}, we proposed the optimal resonant frequencies for DMA beamforming when the target angle is perfectly known. 
			In practice,  the transmitter performs beam training prior to data transmission to determine the beamforming vector based on AoD estimation.
%			 The target angle is generally estimated using a beam training procedure prior to the data transmission stage. 
			 In conventional exhaustive search beam training, the beamforming vector is determined by  sequentially configuring the array for different AoDs and choosing the angle with the highest beamforming gain to compute the optimal beamforming vector for the desired user.
		This process is inefficient for large arrays because the number of time slots required to perform beam training increases with the number of antennas\cite{8240727}.
			In this section, we leverage the frequency reconfigurability of the DMA and propose a  beam training procedure that estimates the optimal  DMA configuration in a single time-slot. 
			%We use an array of DMAs as follows.

				\subsection{System model for an array of DMAs }
				
				Let an array of $N_\sfz$ DMAs  be arranged in the $\sfy\sfz$ plane as shown in Fig.~\ref{fig: dma array diagram}, where each DMA has $N_\sfy$ radiating slots.
				We assume that each DMA has the same intrinsic phase-shift.
				 For a receiver on the $\sfx\sfy$ plane, the effective propagation far-field channel for the array of DMAs is denoted as $	\bh_{\sfA}(\phi,f)\in \bbC^{N_\sfz N_\sfy \times 1}$ and expressed in terms of $\bh(\phi,f)\in \bbC^{N_\sfy  \times 1}$ as
				\begin{align}
						\bh_{\sfA}(\phi,f)= [\underbrace{\bh^T(\phi,f), \bh^T(\phi,f), \dots, \bh^T(\phi,f)}_{N_\sfz \text{ times}}]^T.
				\end{align}
				As the elevation angle $\theta=90^{\circ}$, the propagation channel from each DMA to the receiver in $\sfx\sfy$ plane is the same.
				The extension to do beam training in both azimuth and elevation directions is a topic of future study.

				The beamforming vector for the $m$th DMA in the array is denoted as $\bff_{\mathsf{DMA},m}(f)$.
				  The beamforming vector for the array of DMAs is denoted as
				   $\bff_{\mathsf{ADMA}}(f)$ and  expressed as $\bff_{\mathsf{ADMA}}(f) =[\bff^T_{\mathsf{DMA},1}(f), \bff^T_{\mathsf{DMA},2}(f), \dots , \bff^T_{\mathsf{DMA},N_\sfz}(f)]^T.$
				  We assume that the same signal is input to each DMA and is represented by $	\sfH_{0}(f) $ at the waveguide  input. Hence, the far-field component of the electric field is expressed similar to the single DMA expression from \eqref{eqn: E phi f closed form} as 
				  \begin{align}
				  \sfE_{\phi}(f)	=	\frac{\eta_0}{4\pi r}\left(\frac{2 \pi  f}{c}\right)^2  \sfH_{0}(f) \sfF \sfQ   \bff^T_{\mathsf{ADMA}}(f) \bh_{\sfA}(\phi,f) \left[\frac{\text{V}/\text{m}}{\text{Hz}}\right].
				  \end{align}
%				  $	\sfE_{\phi}(f)	=	\frac{\eta_0}{4\pi r}\left(\frac{2 \pi  f}{c}\right)^2  \sfH_{0}(f) \sfF \sfQ   \bff^T_{\mathsf{ADMA}}(f) \bh_{\sfA}(\phi,f) \left[\frac{\text{V}/\text{m}}{\text{Hz}}\right].$
%				  \begin{align}
%				  \bff_{\mathsf{ADMA}}(f) =[\bff^T_{\mathsf{DMA},1}(f), \bff^T_{\mathsf{DMA},2}(f), \dots , \bff^T_{\mathsf{DMA},N_\sfz}(f)]^T.
%				  \end{align}
				  The frequency-selective beamforming gain  for the DMA array is defined similar to that of a single DMA as 
				  \begin{align}\label{eqn: G dmaa}
				  	\cG_{\mathsf{ADMA}}(\phi, f)\!=  \!\left|\bff^T_{\mathsf{ADMA}}(f) \bh_{\sfA}(\phi,f)\right|^2 \!\!=\!\! \left|\!\sum_{m=1}^{N_\sfz}\! \!\bff^T_{\mathsf{DMA},m}(f)\bh(\phi,f)    \right|^2\!\!\!.
				  \end{align}
				  Similarly, the received power spectral density is $	P_{\mathsf{R}}(\phi,f) = \left(\frac{\lambda}{4 \pi r}\right)^2 P_{\mathsf{T}}(f) 	\cG_{\mathsf{ADMA}}(\phi, f) \left[\frac{\text{W}}{\text{Hz}}\right]$.

				Let the resonant frequency of the $n$th element on the $m$th DMA in the array be denoted as $f_{\sfr, n, m}$. Similar to the analysis in Section~\ref{sec: Frequency-selective beamforming gain optimization and analysis}, the maximum beamforming gain $	\cG^{\star}_{\mathsf{ADMA}}(\phi, f^{\star}_{\sft}(\phi))=N_\sfz^2 N_\sfy^2 $ at a specific angle $\phi$ is attained when resonant frequencies of all elements on all DMAs are set to the optimal operating frequency
				$ f_{\sfr, n, m}=f^{\star}_{\sft}(\phi)$.
				In practice, this optimal resonant frequency needs to be estimated using the beam training procedure as the AoD is unknown. We discuss the beam training procedure to estimate the optimal resonant frequency for the DMA.
%				In practice, the array of DMAs will first estimate the receiver angle and then configure the DMA resonant frequencies based on the estimated  angle. Let the angle estimate be $\hat{\phi}$. Hence, the optimal resonant frequencies to maximize the beamforming gain at the estimated angle and the corresponding optimal operating frequency are $f_{\sfr, n, m}=f^{\star}_{\sft}(\hat{\phi})$. We discuss the procedure to compute $\hat{\phi}$ and then analyze the loss in the beamforming gain because of mismatch between the actual and estimated  angle.
				
			%	\blue{Multiple DMAs arranged in a planar array.}
			\begin{figure}
				\centering
				\includegraphics[width=0.5\textwidth]{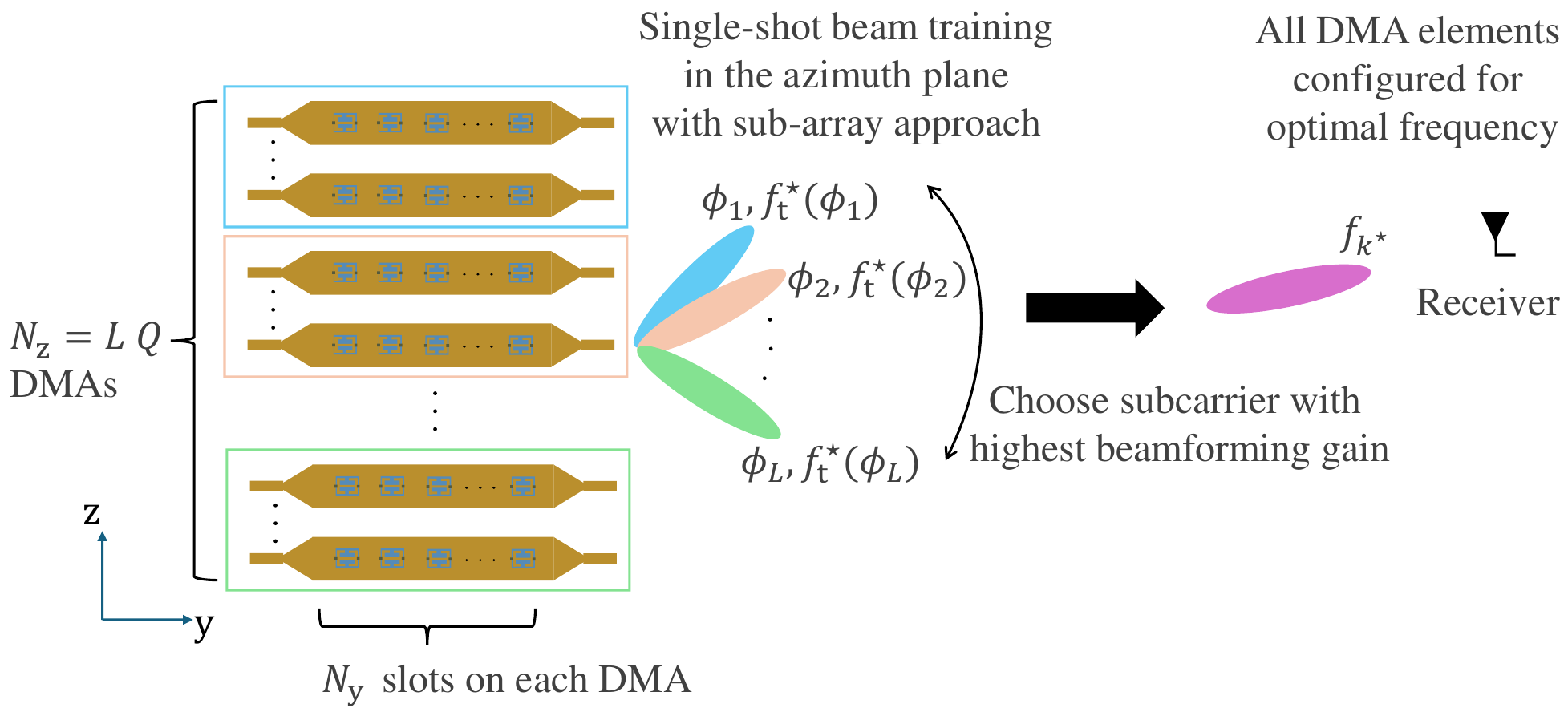}    
				\caption
				{An array of $N_\sfz$ DMAs where each DMA has $N_\sfy$ radiating slots.
					The sub-array approach enables simultaneous probing of different directions using different frequencies to estimate the optimal resonant frequency. Post beam training, all elements of the DMA array are configured for the optimal frequency.
				}
				\label{fig: dma array diagram}
			\end{figure}
				
				\subsection{Beam training with a single OFDM symbol}

				The goal of beam training is to estimate the optimal resonant frequency configuration for the DMA array to serve a user in a specific direction. To avoid training overhead, we propose a single-shot beam training procedure which enables probing different angular directions simultaneously with a single OFDM symbol by leveraging the frequency reconfigurability of the DMA.
			To achieve this goal, 	we use the sub-array approach for beam training. Let the array of  $N_\sfz$ DMAs be grouped into $L$ sub-groups where each sub-group has $Q=\frac{N_\sfz}{L}$ DMAs as shown in Fig.~\ref{fig: dma array diagram}. We assume that the  angle $\phi$ lies inside the angular range $[-\phi_{\mathsf{max}}, \phi_{\mathsf{max}}]$.  We divide this angular range into $L$  sectors where the $\ell$th sector has 
%				an angular range of  $\frac{2\phi_0}{P}$.
its peak gain at angle  $\phi_{\ell}$.
				The DMA array is configured such that the $\ell$th DMA sub-group focuses at  $\phi_{\ell}$ at the corresponding optimal frequency 
				$	f^{\star}_{\sft}(\phi_{\ell})$. The DMA resonant frequencies for  beam training  are configured as 
				\begin{align}\label{eqn: beam training resonant frequencies}
				f_{\sfr, n, m}= f^{\star}_{\sft}(\phi_{\ell}), \quad \forall n,   m \in \{ Q(\ell-1)+1, \dots, Q\ell\} .
				\end{align}
				This approach enables scanning different  angles simultaneously with a single OFDM signal.
				
				The DMA array transmits an OFDM pilot signal with $K_{\mathsf{tr}}$ subcarriers for beam training. Let the $k$th subcarrier frequency be $f_k$. The DMA beamforming vector configured according to \eqref{eqn: beam training resonant frequencies} has the response  $\bff_{\mathsf{ADMA}}(f_k)$ for the $k$th subcarrier.
			Let the actual AoD be $\phi$. The beamforming gain at the receiver at the $k$th subcarrier is $	\cG_{\mathsf{ADMA}}(\phi, f_k)$. The receiver computes the subcarrier index with the highest beamforming gain as 
			\begin{align}\label{eqn: k star}
				k^{\star}=\underset{k}{\mathsf{argmax}}\quad  \cG_{\mathsf{ADMA}}(\phi, f_k).
			\end{align}
			It sends feedback of the optimal subcarrier index to the transmit array.
%			\begin{align}
%				\hat{\phi}=\sin^{-1}\left(\frac{c}{d_\sfy f_{	k^{\star}} }-n_\sfg\right).
%			\end{align}
		For the data transmission stage, the transmit DMA array is configured to operate for a specific bandwidth $\sfB$ centered around $f_{	k^{\star}}$. All DMA resonant frequencies are set to $f_{\sfr, n, m}=f_{	k^{\star}}$.
		% to maximize the beamforming gain at $\hat{\phi}$  and $f_{	k^{\star}}$. 
		With this configuration, the maximum beamforming gain of $N_\sfz^2 N_\sfy^2 $ is attained at $f_{	k^{\star}}$ and  $\hat{\phi}$  where $\hat{\phi}=\sin^{-1}\left(\frac{c}{d_\sfy f_{	k^{\star}} }-n_\sfg\right).$
%		Similar to \eqref{eqn phi c}, the expression for the receiver angle estimate is computed using $k^{\star}$ as
%		$\hat{\phi}=\sin^{-1}\left(\frac{c}{d_\sfy f_{	k^{\star}} }-n_\sfg\right).$

\subsection{Beam training codebook design }

In this section, we propose a codebook for the single-shot beam training approach.
We first analyze the performance loss in beamforming when DMA resonant frequencies are set based on the subcarrier index obtained from the training procedure by simplifying $\cG_{\mathsf{ADMA}}(\phi, f_{k^{\star}})$ as
%We  analyze the beamforming gain at $\phi$, $	\cG_{\mathsf{ADMA}}(\phi, f_k)$, when the DMA array is configured to maximize the beamforming gain at the optimal subcarrier frequency $f_{	k^{\star}}$ and the corresponding  angle estimate $\hat{\phi}$.
%Similar to the beamforming gain  for a single DMA in \eqref{eqn:  G_dma_phit_f_eqn_in_proof}, we have 
%\begin{align}\label{eqn: G DMAA phit f closed form}
%	\cG_{\mathsf{ADMA}}(\phi, f_k) &=N_\sfz^2\left(\frac{\sin\left( \pi  N_\sfy\frac{ f_k d_\sfy^{\star} (n_\sfg^{\star} +\sin(\phi))}{c} \right)}{\sin\left( \pi \frac{ f_k d_\sfy^{\star} (n_\sfg^{\star} +\sin(\phi))}{c} \right)}\right)^2 \nonumber \\&\times\frac{\Gamma^2 f_k^2}{(2\pi f_{k^{\star}}^2- 2\pi f_k^2 )^2+ \Gamma^2 f_k^2}.
%\end{align}
%At $f_k=f_{k^{\star}}$, the beamforming gain simplifies to 
%$	\cG_{\mathsf{ADMA}}(\phi, f_{k^{\star}}) =\left(\frac{\sin\left( \pi N_\sfz N_\sfy\frac{  (n_\sfg^{\star} +\sin(\phi))}{(n_\sfg^{\star} +\sin(\hat{\phi}_{\sft}))} \right)}{\sin\left( \pi \frac{  (n_\sfg^{\star} +\sin(\phi))}{(n_\sfg^{\star} +\sin(\hat{\phi}_{\sft})))} \right)}\right)^2 $.
				\begin{align}\label{eqn: Gdmaa in terms of phit hat}
					\cG_{\mathsf{ADMA}}(\phi, f_{k^{\star}}) =N_\sfz^2\left(\frac{\sin\left( \pi  N_\sfy\frac{  (n_\sfg^{\star} +\sin(\phi))}{(n_\sfg^{\star} +\sin(\hat{\phi}))} \right)}{\sin\left( \pi \frac{  (n_\sfg^{\star} +\sin(\phi))}{(n_\sfg^{\star} +\sin(\hat{\phi}))} \right)}\right)^2 .
				\end{align}
	From  \eqref{eqn: Gdmaa in terms of phit hat}, we see that the maximum beamforming gain of $N_\sfz^2 N_\sfy^2$ is attained when $\hat{\phi}= \phi $.
	As the error between the actual and estimated  angle increases, the beamforming gain depreciates.			In the following lemma, we quantify the allowable range of  angle estimates for which the beamforming gain is above  a certain threshold of $\delta$ times the maximum gain.		
		Let $\Psi=\frac{  (n_\sfg^{\star} +\sin(\phi))}{(n_\sfg^{\star} +\sin(\hat{\phi}))}$.
		The maximum value of $\left(\frac{\sin(\pi N_\sfy \Psi)}{\sin(\pi \Psi)}\right)^2$ is $N_\sfy^2 $.
			We define 
	$\Psi_\delta( N_\sfy)$ such that for $\Psi \in [-\Psi_\delta( N_\sfy), \Psi_\delta( N_\sfy)]$, the term $\left(\frac{\sin(\pi N_\sfy \Psi)}{\sin(\pi \Psi)}\right)^2 \geq { \delta N_\sfy^2}$ for $\delta <1$ We can compute $\Psi_\delta( N_\sfy)$ numerically. For example, for $\delta=0.5$ (3 dB from maximum), we have $\Psi_\delta( N_\sfy) \approx \frac{0.448}{N_\sfy}$.
				\begin{lemma}\label{thm: GDMAA gain threshold}
				The inequality $	\cG_{\mathsf{ADMA}}(\phi, f_{k^{\star}}) \geq {\delta N_\sfy^2 N_\sfz^2},$
%					\begin{align}\label{eqn: GDMAA threshold}
%					\cG_{\mathsf{ADMA}}(\phi, f_{k^{\star}}) \geq \frac{N_\sfy^2 N_\sfz^2}{10^{{\delta/10}}},
%					\end{align}
is satisfied when
					the angle estimate $\hat{\phi}$ is bounded as follows.
					\begin{align}\label{eqn: phi t hat bounds}
						\hat{\phi}\! \in \!\! \left[\sin^{-1}\!\left(\!\frac{n_\sfg^{\star}+\sin(\phi)}{1+\Psi_\delta( N_\sfy)}-n_\sfg^{\star}\!\right)\!, \sin^{-1}\!\left(\!\frac{n_\sfg^{\star}+\sin(\phi)}{1-\Psi_\delta( N_\sfy)}-n_\sfg^{\star}\!\right)\! \right].
					\end{align}
				\end{lemma}
				\begin{proof}
			 When $\hat{\phi}=\phi $, we have $\Psi=1$.  To satisfy the gain threshold, we have  $	1-\Psi_\delta( N_\sfy) \leq \frac{  (n_\sfg^{\star} +\sin(\phi))}{(n_\sfg^{\star} +\sin(\hat{\phi}))} \leq 1+\Psi_\delta( N_\sfy).$
%				\begin{align}\label{eqn: condition on psi}
%					1-\Omega(\delta, N_\sfy) \leq \frac{  (n_\sfg^{\star} +\sin(\phi))}{(n_\sfg^{\star} +\sin(\hat{\phi}_{\sft}))} \leq 1+\Omega(\delta, N_\sfy).
%				\end{align}
				Simplifying, we obtain the  bounds on $	\hat{\phi}$ in  \eqref{eqn: phi t hat bounds}.
				\end{proof}
				
				\begin{figure}
					\centering
					\includegraphics[width=0.4\textwidth]{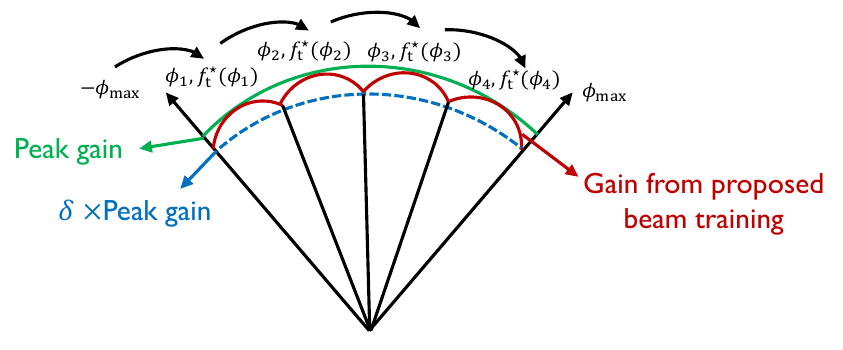}    
					\caption
					{The proposed beam codebook design ensures that the beamforming gain for all angles in the desired range is within a factor of $\delta$ from the maximum where $\delta$ is a design parameter.				
					}
					\label{fig: beam codebook design}
				\end{figure}
				
					\begin{figure}
					\centering
					\begin{subfigure}[t]{0.49\linewidth}
						\centering
						\includegraphics[width=1.1\linewidth]{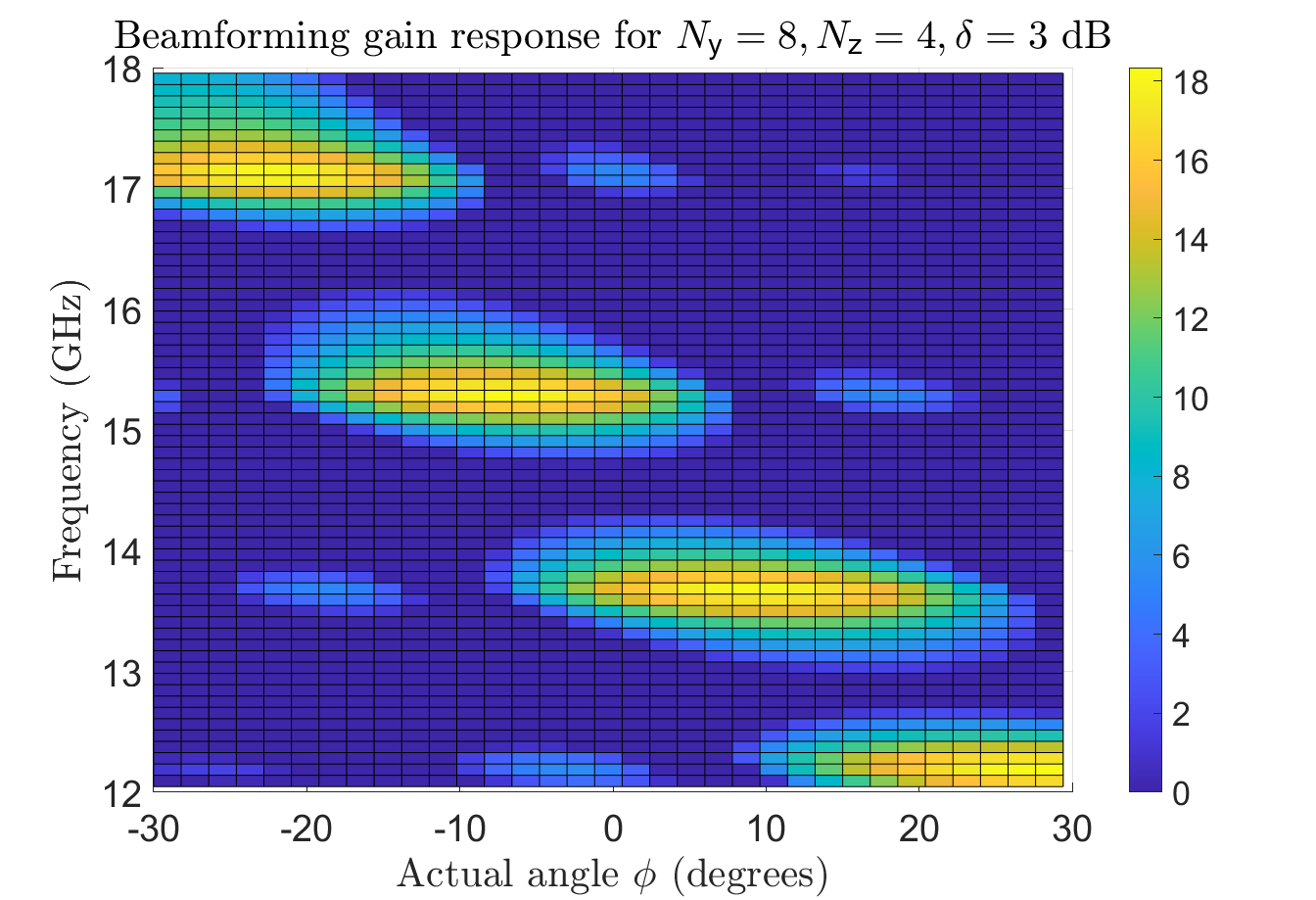}    
						\caption
						{Beam training with $N_\sfy=8$, $N_\sfz=4$, $\delta=3$ dB
						}
						\label{fig: bf gain Ny8 Nz4}
					\end{subfigure}~\hfil
					\begin{subfigure}[t]{0.49\linewidth}
						\centering
						\includegraphics[width=1.1\linewidth]{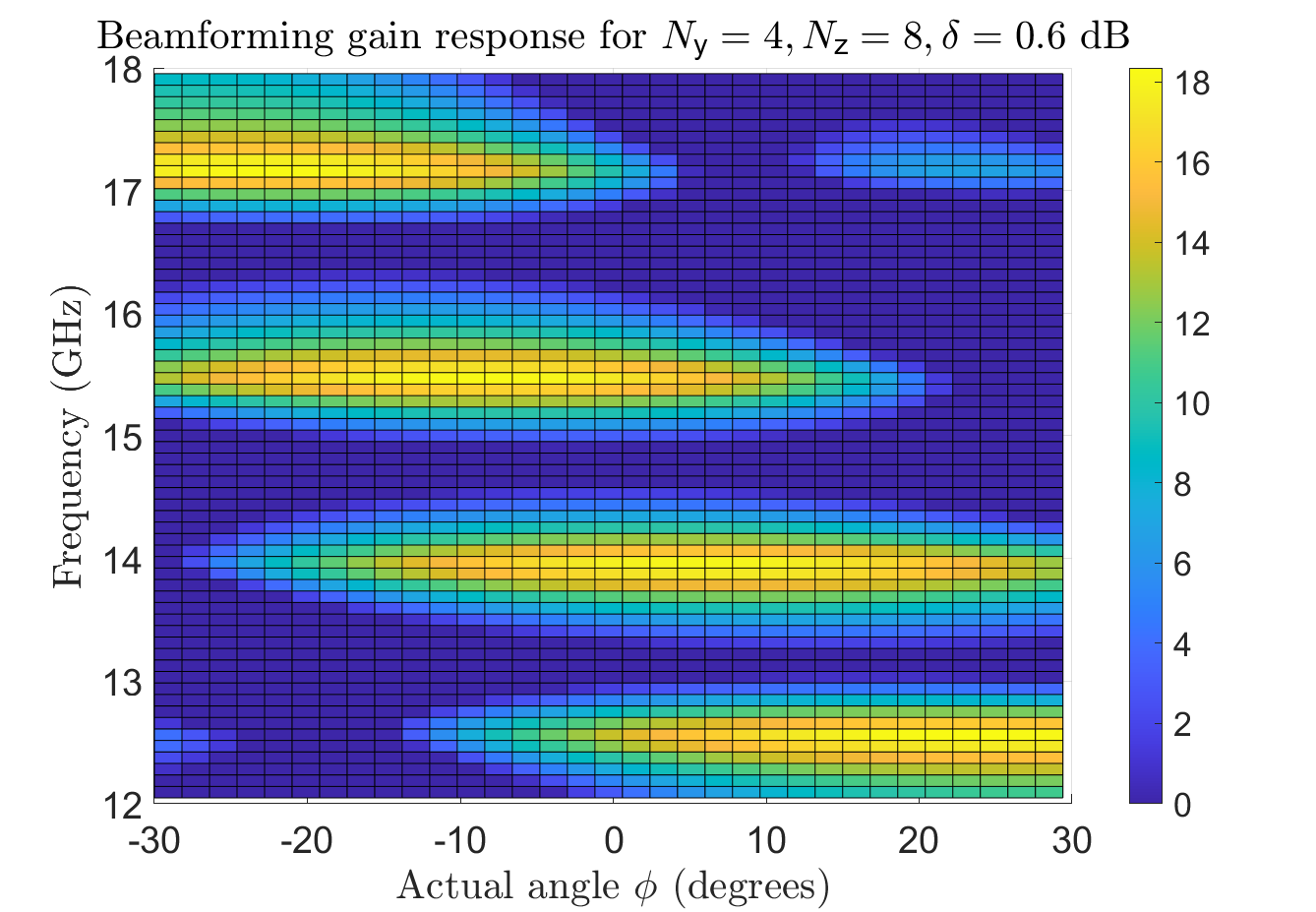}    
						\caption
						{Beam training with $N_\sfy=4$, $N_\sfz=8$, $\delta=0.6$ dB
						}
						\label{fig: bf gain Ny4 Nz8}
					\end{subfigure}
					\caption{The optimal operating frequency decreases in $L=4$ steps as a function of the actual angle similar to the monotonic decreasing trend obtained for the continuous curve in Fig.~\ref{fig: optimal_target_freq_target_angle}. For the same $L$, the angular beam width is narrower for $N_\sfy=8$ compared to $N_\sfy=4$.  }
				\end{figure}

		\begin{figure}
	\centering
	\includegraphics[width=0.4\textwidth]{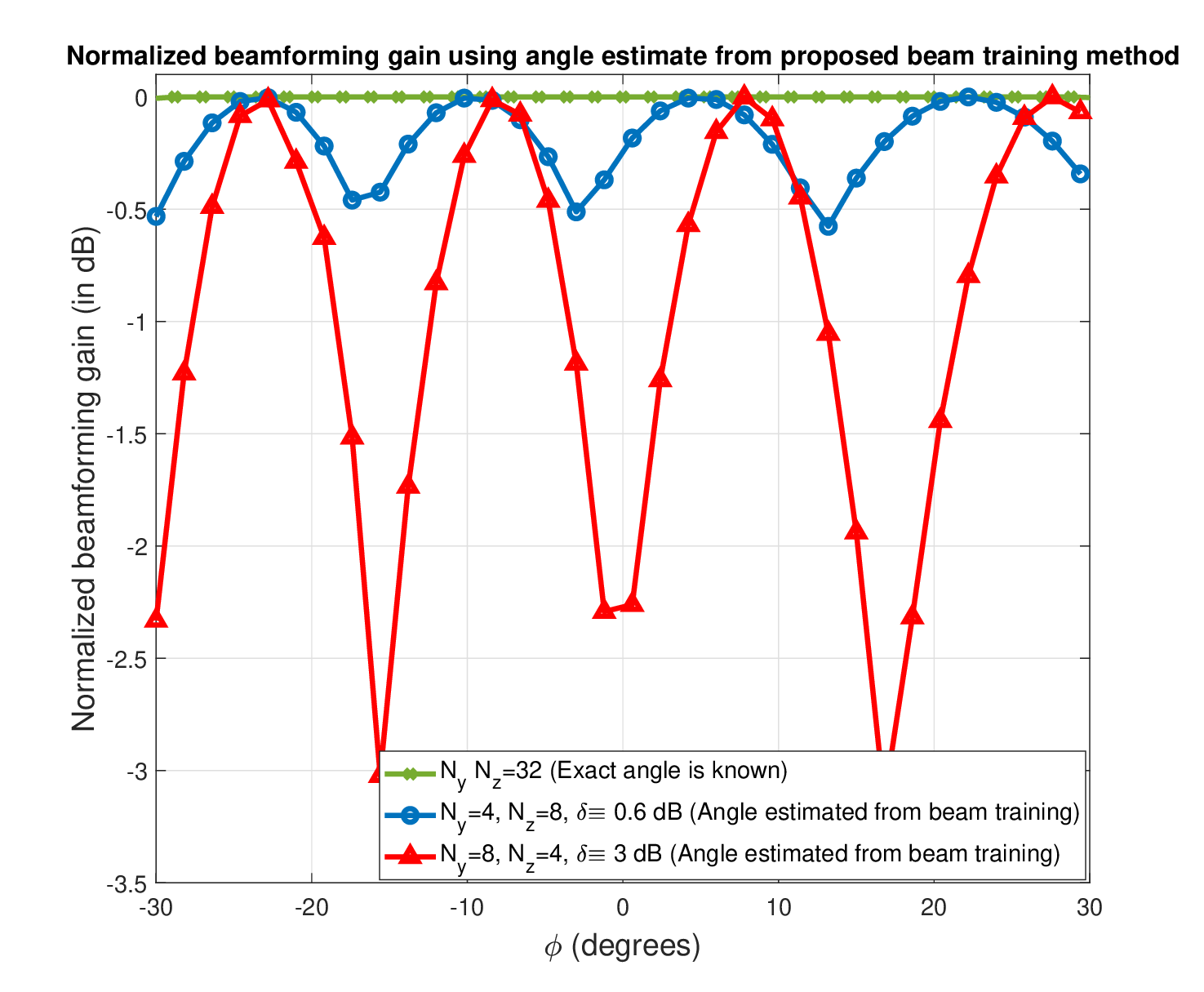}    
	\caption
	{Beamforming gain at the optimal operating frequency obtained using  angle estimate depends on $N_\sfy$ and $N_\sfz$ whereas beamforming gain with perfect angle depends solely on the product $N_\sfy N_\sfz$.				
	}
	\label{fig: Estimated target angle and optimal target frequency}
\end{figure}

				\textbf{ Beam codebook design:} 
				The goal of beam codebook design is to choose the angles where peak gain occurs for the $L$ sectors in the beam training procedure to ensure that the beamforming gain $\cG_{\mathsf{ADMA}}(\phi, f_{k^{\star}})$ for any $\phi$ in the desired range does not drop below a certain threshold of $\delta$ from the peak gain. 
				We use the result from Lemma~\ref{thm: GDMAA gain threshold} to design the beam codebook.
%				Using the result from Lemma~\ref{thm: GDMAA gain threshold}, we design the  angles where peak gain occurs for the $L$ sectors in the beam training procedure to ensure that the beamforming gain  does not drop below a certain threshold of $\delta$ from the peak gain.
				 For the $\ell$th sector, let $\phi_\ell$ be the angle with the peak gain, i.e., $\cG_{\mathsf{ADMA}}(\phi_\ell, f_{k^{\star}}) =N_\sfy^2 N_\sfz^2 ~\forall \ell=1, \dots, L$. 
				The main objective of the beam codebook design is to ensure that the beamforming gain for all angles in the range $[-\phi_{\mathsf{max}}, \phi_{\mathsf{max}}]$ is above $\delta N_\sfy^2 N_\sfz^2$.
				The peak angle of the first sector is designed such that the gain at  $-\phi_{\mathsf{max}}$ is $\delta N_\sfy^2 N_\sfz^2$. Equating the lower limit of $\hat{\phi}$ when $\phi=\phi_1$ in \eqref{eqn: phi t hat bounds}  to $-\phi_{\mathsf{max}}$, we obtain the expression of $\phi_1$ as
				\begin{equation}\label{eqn: phi1}
					\phi_{1}=\sin^{-1}\left((\sin(-\phi_{\mathsf{max}})+n_\sfg^{\star}) (1+\Psi_\delta( N_\sfy))-n_\sfg^{\star}\right).
				\end{equation}
				For the remaining sectors, we equate the lower limit of  $\hat{\phi}$  when $\phi=\phi_{\ell}$ to the upper limit when $\phi=\phi_{\ell-1}$. With this procedure, the remaining $L-1$ angles are obtained sequentially as 
					\begin{equation}\label{eqn: phi ;}
					\phi_{\ell}\!=\!\sin^{-1}\!\!\left(\!(\sin(\phi_{\ell-1})+n_\sfg^{\star}) \frac{(1+\Psi_\delta( N_\sfy))}{(1-\Psi_\delta( N_\sfy))}-n_\sfg^{\star}\!\right),\!  \forall \ell>1.
				\end{equation}
	The illustration of the beam codebook design procedure is shown in Fig.~\ref{fig: beam codebook design}.
	To compare differences in the arrangement of elements on the DMA, and the impact on beam training,
		we describe two example cases when $\phi_{\mathsf{max}}=30^{\circ}$ and $n_\sfg^{\star}=2.5$ as follows:
		%For example, when $\phi_0=30^{\circ}$ and $n_\sfg^{\star}=2.5$, 
		\begin{itemize}
			\item For $N_\sfy=8$, $N_\sfz=4$, $\delta\equiv 3$ dB, we have $	\phi_{1}=-22.83^{\circ}$, $	\phi_{2}=-7.9^{\circ}$, $	\phi_{3}=8.21^{\circ}$, $	\phi_{4}=27.16^{\circ}$.
			\item  For $N_\sfy\!=4$, $N_\sfz\!=8$, $\delta\! \equiv 0.6$ dB, we have $	\phi_{1}=-23.33^{\circ}$, $	\phi_{2}=-9.51^{\circ}$, $	\phi_{3}=5.22^{\circ}$, $	\phi_{4}=22.04^{\circ}$.  
		\end{itemize}
		
		The DMA resonant frequencies are set using \eqref{eqn: beam training resonant frequencies}.
		The beamforming gain response as a function of frequency and angle for these two cases is 
		  shown in Fig.~\ref{fig: bf gain Ny8 Nz4} and Fig.~\ref{fig: bf gain Ny4 Nz8}.
%		As $L=4$, we obtain four distinct values of the  optimal subcarrier frequency as shown in Fig.~\ref{fig: Optimal target frequency}.  The error in the angular estimates is shown in Fig.~\ref{fig: Estimated target angle}.
The optimal subcarrier index at the receiver is computed using \eqref{eqn: k star}.
	Using the optimal resonant configuration, we evaluate \eqref{eqn: Gdmaa in terms of phit hat} and plot the normalized beamforming gain at the optimal operating frequency as a function of the actual  angle in Fig.~\ref{fig: Estimated target angle and optimal target frequency}. From the plot, we verify that the beamforming gain variation is within the desired threshold. This validates the result in Lemma~\ref{thm: GDMAA gain threshold}. Using only four distinct values of resonant frequencies, it is possible to configure the DMA array such that the beamforming gain is within a certain threshold of the theoretical maximum obtained with perfect knowledge of the AoD.
		It suffices to use a DMA architecture with varactor diodes having the ability to be tuned to four distinct states instead of a continuous high resolution tunability.

				\section{Achievable rate results }\label{subsec: Achievable rate results}
				
				We evaluate the achievable rate performance of the DMA array using the results from Section~\ref{sec: Frequency-selective beamforming gain optimization and analysis} and Section~\ref{sec: Single-shot beam training with DMAs}.
%				The communication bandwidth is $\sfB$.  Let the center frequency of the data transmission band be $f_{\sfd}$.
%				The DMA array transfers data in the band $[f_{\sfd}- \frac{\sfB}{2}, f_{\sfd}+\frac{\sfB}{2}]$.
%				In the data transmission band, assume an OFDM system with $K_{\sfd}$ subcarriers. 				
%				The achievable rate is defined as $\mathsf{R}(\phi)=\frac{\sfB}{K_{\sfd}}\sum_{k=1}^{K_{\sfd} }\log_2\left(1+\mathsf{SNR}(\phi,f_k)\right) \left[\mathsf{bits/s}\right],$
%%				\begin{align}
%%					\mathsf{R}(\phi)=\frac{\sfB}{K_{\sfd}}\sum_{k=1}^{K_{\sfd} }\log_2\left(1+\mathsf{SNR}(\phi,f_k)\right) \left[\mathsf{bits/s}\right],
%%				\end{align}
%				where $\mathsf{SNR}(\phi,f)$ is defined for the DMA array similar to that of a single DMA in \eqref{eqn: snr}.
%				 We assume equal power allocation, i.e., $	P_{\mathsf{T}}(f) =\frac{\sfP}{\sfB}$. 
In our proposed approach, the center frequency $f_{\sft}$ is dynamically reconfigured  depending on the AoD. We assume two cases for comparison:
					\textbf{Case 1-} When the AoD is perfectly known, the center frequency is set  using Lemma~\ref{thm:  optimal target freq};
					\textbf{Case 2-} When the AoD is unknown, we use the single-shot beam training procedure from Section~\ref{sec: Single-shot beam training with DMAs} and set  $f_{\sft}=f_{k^{\star}}$ obtained from \eqref{eqn: k star}.
				We also compare these two cases from our proposed approach to that of the benchmark DMA approach based on a fixed operating frequency.   For the benchmark  case, we set $f_{\sft}=f_{\sfc}$ which is the center frequency of the DMA feasible operating frequency range.
				To compare with TTD, we also plot the achievable rate for TTD in  Fig.~\ref{fig:rate  with target angle estimates}.
				
				\begin{figure}
					\centering
					\includegraphics[width=0.45\textwidth]{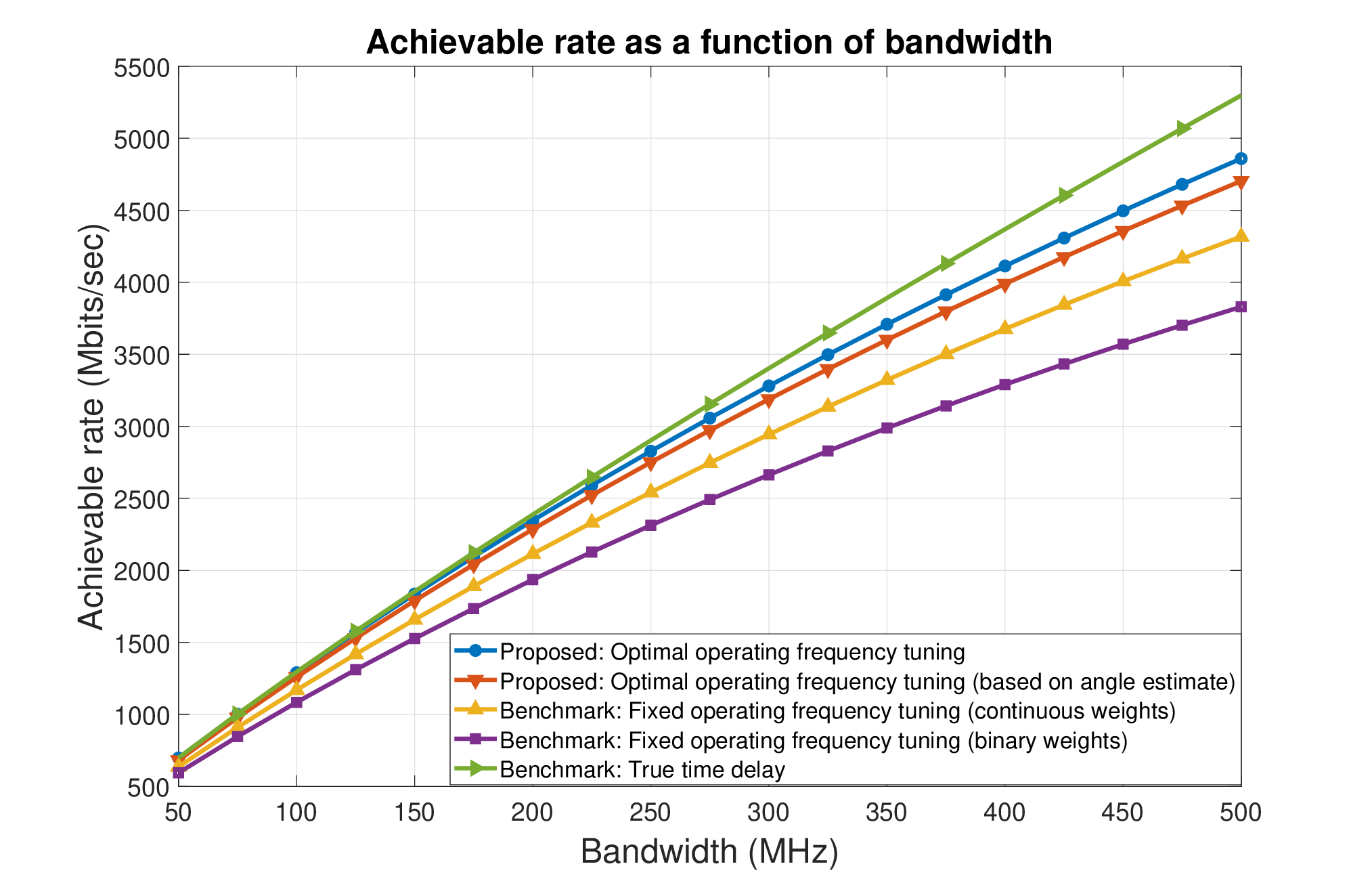}    
					\caption
					{The achievable rate DMA performance is comparable to the TTD  performance up to a bandwidth of 300 MHz which is also the 3 dB DMA bandwidth as shown in Fig.~\ref{fig: Beamforming_gain_freq_response_3dB}.
					}
					\label{fig:rate  with target angle estimates}
				\end{figure}
				\begin{figure}
					\centering
					\includegraphics[width=0.45\textwidth]{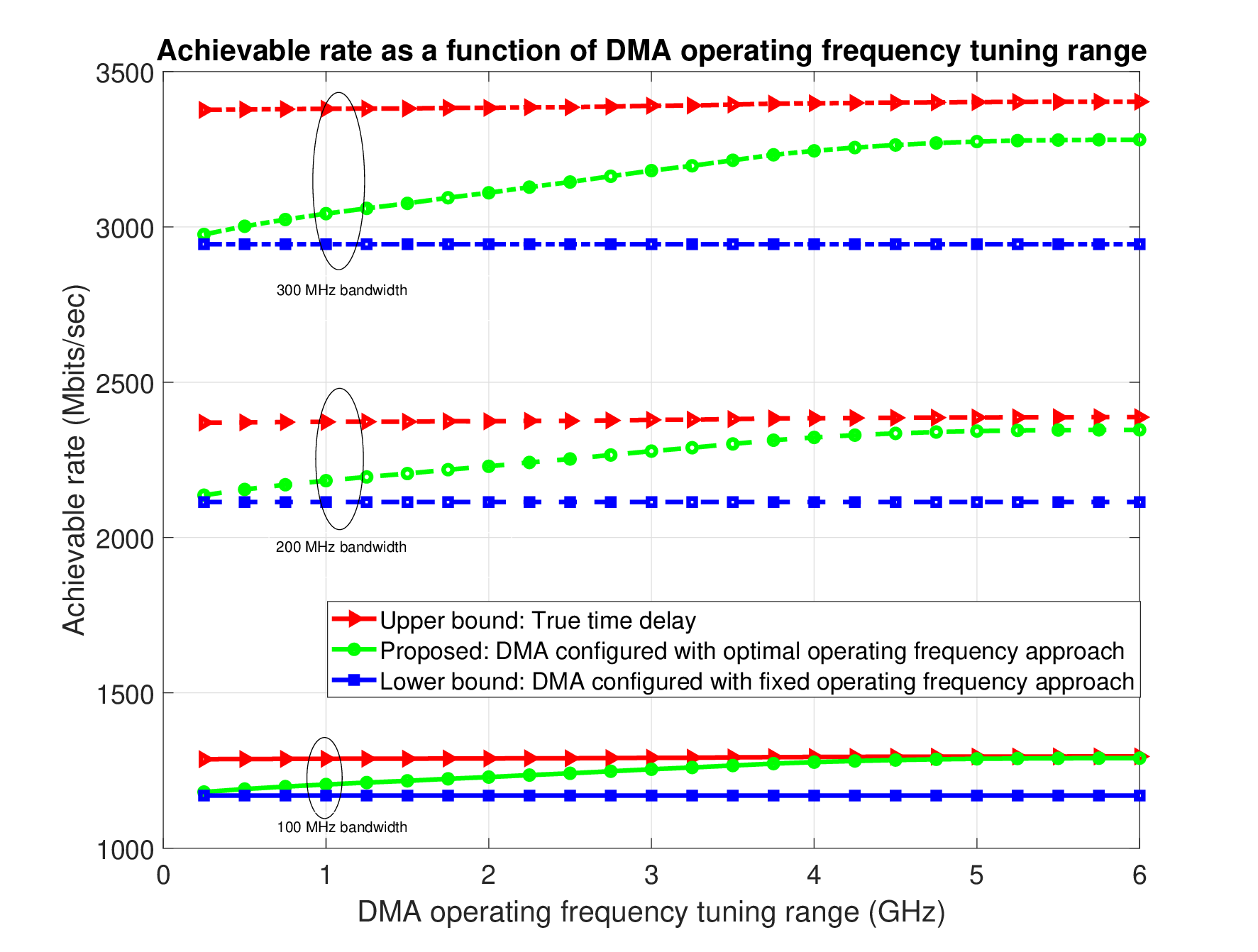}    
					\caption{
						DMA achievable rate based on optimal operating frequency approach is lower bounded by the DMA rate based on fixed frequency and upper bounded by the TTD rate for all tuning ranges.
					}
					\label{fig:rate  with vary tr}
				\end{figure}
				
				For  Fig.~\ref{fig:rate  with target angle estimates}, we set $\sfP=250$ mW, $r=500$ m, $n_\sfg=2.5$, $f_{\sfc}=15$ GHz, $d_\sfy=0.42 \lambda_{\sfc}$, $\Gamma=\frac{2\pi f_{\sfc}}{50}$, $N_\sfy=8$, $N_\sfz=4$, and $\sfT_{\sfr}=6$ GHz. For beam training, we use $L=4$ and the angles with peak gain for each sector are chosen corresponding to the threshold $\delta\equiv3$ dB.
				In Fig.~\ref{fig:rate  with target angle estimates},  we plot the achievable rate averaged over  angles $\phi\sim \cU(-30^{\circ}, 30^{\circ})$.  In the OFDM system, we have $K_{\sfd}=64$ subcarriers. On the $\sfx$ axis of Fig.~\ref{fig:rate  with target angle estimates}, we vary the communication bandwidth $\sfB$.
				We observe that the achievable rate of the DMA configured based on our proposed optimal operating frequency approach outperforms the benchmark approach based on the fixed operating frequency. We also observe that the rate performance of the DMA when resonant frequencies are configured based on the  single-shot beam training procedure is comparable to the rate obtained using perfect knowledge of the AoD. This validates the efficacy of the proposed single-shot beam training approach. 
				 We also showed in Fig.~\ref{fig: Estimated target angle and optimal target frequency} that with $L=4$ sectors, the beamforming gain varies within 3 dB of the maximum.
%				  Hence, it suffices to use a data converter with just two bit resolution to configure the varactor diodes on the DMA.
%				   This simple two bit  architecture for varactor tuning can be used to do single-shot beam training and also effectively communicating data using only four resonant frequency configurations.
Hence, using only four resonant frequency configurations, it is possible to do both beam training and beamforming.
				In Fig.~\ref{fig:rate  with target angle estimates}, we see that the achievable rate increases with bandwidth. The DMA achievable rate is comparable to that of TTD  up to a bandwidth of 300 MHz (DMA 3 dB bandwidth) beyond which the performance gap increases. This is because of the   DMA beamforming gain frequency-selectivity as shown in Fig.~\ref{fig: Beamforming_gain_freq_response_3dB}.

	In Fig.~\ref{fig:rate  with target angle estimates}, we assumed that the DMA operating range $\sfT_{\sfr}$ is 6 GHz. In practice, this range depends on the reconfigurability of the DMA. In Fig.~\ref{fig:rate  with vary tr}, we vary $\sfT_{\sfr}$ by keeping all parameters same as Fig.~\ref{fig:rate  with target angle estimates} and plot the achievable rate obtained with three different values of the bandwidth $\sfB$. The baseline approach performance does not change with $\sfT_{\sfr}$ as the DMA operating frequency is always fixed to $f_{\sfc}$. We observe that the proposed approach of dynamically configuring the DMA operating frequency  outperforms the baseline fixed operating frequency approach for all tuning ranges. As $\sfT_{\sfr}$ increases, the DMA achievable rate   from the proposed approach tends to the TTD rate which acts as an upper bound.
	This implies that a DMA with high tuning range can be an alternative to the  power hungry TTD architecture.

				\section{Conclusion and future work}\label{sec: Conclusion and future work }
				
			In this paper, we proposed a novel approach where the system  configures the  DMA resonant frequencies by optimally choosing the  operating frequency based on the AoD.  We derived a closed-form  expression for the beamforming gain frequency response obtained using our proposed  approach.  We also developed a single-shot beam training approach for the DMA to estimate the optimal resonant frequency configuration. This approach gives comparable achievable rate performance  with that obtained using perfect knowledge of the AoD.
%			The efficacy of the beam training approach is validated by comparing the achievable rate performance with that obtained using perfect knowledge of the receiver angle.
  Our  approach for configuring the DMA outperforms the benchmark fixed operating frequency approach for all tuning ranges. 
		%	The achievable rate of TTD and DMA are comparable upto a certain fractional bandwidth.
%			The rate performance is also comparable with the rate of TTD  upto a certain fractional bandwidth.  
			For future work, we plan to investigate the performance of the proposed approach in more practical settings like multipath wireless channel, incorporating mutual coupling between elements~\cite{ramirez2022performance},   and impedance matching impairments~\cite{10437975}.
			We also plan to extend the DMA beamforming approach  to a multi-user scenario where reconfigurability of DMAs can be used for interference cancellation. 
%		Another potential direction of future study is extending the	frequency-selective beamforming approach  to a multi-user scenario where the reconfigurability of DMAs can be used for interference steering.

						\bibliographystyle{IEEEtran}
					\bibliography{references_Nitish.bib}

					\end{document}